\documentclass{siamltex}

\usepackage{amsmath, amssymb, amsfonts}
\usepackage{xspace}
\usepackage{pdfsync, color, verbatim}%

\newcommand\R{\mathbb R} 
\newcommand\C{\mathbb C}

\newcommand{\image}{\operatorname{Im}}
\newcommand{\adj}{\operatorname{adj}}

\newcommand{\Psbc}{P_{+ \cdot \cdot}}

\newcommand{\Pasc}{P_{\cdot + \cdot}}

\newcommand{\Pabs}{P_{\cdot \cdot +}}

\newcommand{\hyper}{\Delta(P)}
\newcommand{\ones}{\mathbf 1}
\newcommand{\tr}{\mathrm{tr}}
\newcommand{\Diag}{\operatorname{Diag}}
\newcommand{\Cov}{\operatorname{Cov}}
\newcommand{\prob}{\operatorname{Prob}}
\newcommand{\GL}{\operatorname{GL}}
\newcommand{\Flat}{\operatorname{Flat}}
\newcommand{\GMk}{$\operatorname{GM}(k)$\xspace}
\newcommand{\GMtwo}{$\operatorname{GM}(2)$\xspace}

\title{A semialgebraic description of the general Markov model on
  phylogenetic trees}

\author{Elizabeth S.~Allman\footnotemark[1]\ \footnotemark[2]
\and John A.~Rhodes\footnotemark[1]\ \footnotemark[2]
\and Amelia Taylor\footnotemark[3]}

\begin{document}
\maketitle

\renewcommand{\thefootnote}{\fnsymbol{footnote}}

\footnotetext[1]{Department of Mathematics and Statistics, 
University of Alaska, Fairbanks AK 99775}
\footnotetext[3]{Department of Mathematics, Colorado College, Colorado Springs 80903}
\footnotetext[2]{The work of ESA and JAR was supported by a grant from the National
Science Foundation, DMS 0954865.}

\renewcommand{\thefootnote}{\arabic{footnote}}

\begin{abstract}
Many of the stochastic models used in inference of phylogenetic trees
from biological sequence data have polynomial parameterization maps. The image of 
such a map --- the collection of joint distributions for a model --- forms the model space.  
Since the parameterization is polynomial, the Zariski closure of the model space is
an algebraic variety which is typically much larger than the model space, but has been 
usefully studied with algebraic methods.  Of ultimate interest, however, is not the full variety, 
but only the model space.  Here we develop complete semialgebraic descriptions of the model space 
arising from the $k$-state general Markov model on a tree, with slightly restricted parameters.
Our approach depends upon both recently-formulated analogs of Cayley's hyperdeterminant, and 
the construction of certain quadratic forms from the joint distribution whose positive (semi-)definiteness 
encodes information about parameter values.
We additionally investigate the use of Sturm sequences for obtaining similar results.
\end{abstract}



\pagestyle{myheadings}
\thispagestyle{plain}
\markboth{Allman, Rhodes, and Taylor}{Semialgebraic description of the GM model on phylogenetic trees}

\section{Introduction} Statistical inference of evolutionary relationships among 
organisms from DNA sequence data is routinely
performed using probabilistic models of sequence evolution along a tree. A site in a 
sequence is viewed as a 4-state (A,C,G, T) random variable, which undergoes state changes
as it descends along the tree  from an ancestral organism to its modern descendants. 
Such models exhibit a rich 
mathematical structure, which reflects both the combinatorial features of the tree, and the
algebraic way in which stochastic matrices associated to edges of the
tree are combined to produce a joint probability distribution describing sequences of the extant organisms.  

One thread in the extensive literature on such models has utilized the viewpoint of algebraic 
geometry to understand the probability distributions that may arise.
This is natural, since the distributions are in the image of a polynomial map, and the 
image thus lies in an algebraic variety. The defining equations of this
variety  (which depend on the tree topology), are called \emph{phylogenetic invariants}. 
That a probability distribution satisfies them can be taken as evidence
that it arose from sequence evolution along the particular tree.
Phylogenetic invariants and varieties have been extensively studied by many authors 
\cite{CF87, Lake87, EvansSpeed93, Hendy1996, Hag2000, AR03, SStoric2005, ARgm, CFS2011, RSlargeMixtures2012, Buc2012} (see \cite{ARgascuelchapter} for more references) with
goals ranging from biological (improving data analysis) to statistical
(establishing the identifiability of model parameters) to purely
mathematical.  

However, it has long been understood that, in addition to the equalities of phylogenetic 
invariants, inequalities should play a
role in characterizing those distributions
actually of interest for statistical purposes.   Much of a phylogenetic variety is typically
composed of points not arising from stochastic parameters, but rather
from applying the same polynomial parameterization map to complex
parameters. 
Thus the model space --- the set of probability distributions arising as the image of
stochastic parameters on a tree --- can be considerably smaller 
than the set of all probability distributions on the variety.
An instructive recent computation \cite{SZ2011} demonstrated that for the $2$-state
general Markov model on the $3$-leaf tree, for example, the model space is
only about 8\% of the nonnegative real points on the variety. Inequalities can thus be 
crucial in determining if a probability distribution arises from a model.

In the pioneering 1987 paper of Cavender and Felsenstein \cite{CF87} polynomial 
equalities and inequalities are given that can test
which of the 3 possible unrooted leaf-labeled 4-leaf trees might have produced a given probability distribution, and 
thus in principle determine evolutionary relationships between 4 organisms.
Nonetheless, despite many advances in understanding phylogenetic invariants in the intervening years, 
little has been accomplished in finding or understanding the necessary inequalities. 
The potential usefulness of such inequalities, meanwhile,
has been demonstrated in \cite{FallerSteel}, where an inequality that holds for the 
2-state model on all tree topologies plays a key role in studying loss of biodiversity 
under species extinction. In \cite{ARS2012}  a small number of inequalities, dependent on 
the tree, were used to show that for certain mixture models trees were identifiable from probability distributions.

Recent independent  works by Smith and Zwiernik \cite{SZ2011} and  by Klaere and
Liebscher \cite{KL2011}  provided the first substantial progress on the general problem of 
finding sufficient inequalities to describe the model space.
Both groups
successfully formulated inequalities for 
the 2-state general Markov model on trees, using different viewpoints.
While the 2-state
model has some applicability to DNA sequences, through a
purine/pyrimidine encoding of nucleotides, it is unfortunately not clear how to
extend these works to the more general $k$-state model, or even to the particular $k=4$ 
model that is directly applicable to
DNA. Features of the statistical framework in \cite{SZ2011}
make generalizing to more states highly problematic, while the
formulation of \cite{KL2011} involves beating through a thicket of
algebraic details which are similarly difficult to generalize.

\smallskip

In this work we provide a third approach to understanding the model
space of the general Markov model on trees which
has the advantage of extending from the 2-state to the $k$-state model with little
modification.  Our goal is a semialgebraic description 
(given by a boolean combination of polynomial equalities and inequalities) 
of the set of probability distributions that arise on a specific tree. Such a description exists
by the 
Tarski-Seidenberg Theorem \cite{Tarski1951, Seidenberg1954},
since the stochastic 
parameter space for any $k$-state general Markov model is a semialgebraic set, so
its image under the polynomial parameterization map must be as well.
However, we seek an explicit description, and this theorem does not provide a useful means of obtaining it.

We describe below two methods for obtaining such a semialgebraic model description.
In one approach, that applies equally easily to all $k$ and all binary trees, we obtain inequalities
using a recently-formulated analog of Cayley's $2\times2\times 2$ hyperdeterminant from 
\cite{AJRS2012}, and the construction of certain quadratic forms from the joint distribution 
whose positive (semi-)definiteness encodes information about parameter values. We note 
that the appearance of the hyperdeterminant in both  \cite{SZ2011} 
 and \cite{KL2011} motivated the work of \cite{AJRS2012}, but that our introduction of quadratic 
 forms in this paper is an equally essential tool for obtaining our results. Moreover, we do not 
 see direct precursors of this idea in either \cite{SZ2011} or \cite{KL2011}. 

We also describe an alternative method using
Sturm sequences for univariate polynomials to obtain inequalities.
Specifically, we construct polynomials in the entries of a probability distribution 
whose roots are exactly a subset of the
numerical parameters, and Sturm theory leads to inequalities stating that the 
roots lie in the interval $(0, 1)$, as the parameters must. Although for the $2$-state 
model this leads to a complete semialgebraic description of the model on a 
$3$-leaf tree, for higher $k$ it  becomes more unwieldy. Nonetheless, this 
approach can produce inequalities of smaller degree than those found using 
quadratic forms, so we consider it a potentially useful technique.

In both approaches, we must impose some restrictions on the set of stochastic
parameters in order to give our semialgebraic conditions.  We thus formulate a notion of 
\emph{nonsingular parameters} and mostly restrict to considering them for our results.
In the $k=2$ case this notion is particularly natural from a statistical point of view, though 
it is slightly less so for higher $k$. Indeed, an understanding of why this notion is needed
algebraically illuminates, we believe, the difficulties of passing from 2-state results to $k$-state results.

\smallskip

This paper is organized as follows:  In \S \ref{sec:defns} we formally introduce the
general Markov model on trees and set basic notations and terminology, including the notion of nonsingular parameters.
In \S \ref{sec:THREEleaf}, we give a semialgebraic description of the 
general Markov model on the $3$-leaf tree using the work of \cite{AJRS2012} and Sylvester's theorem
on quadratic forms, a description that is made complete
for the $2$-state model, but holds only for nonsingular parameters of the $k$-state model.
Additionally, we discuss connections to several previous works on the $2$-state 
model \cite{PT86, SZ2011, KL2011}.
In \S \ref{sec:Sturm}, we use Sturm sequences to give partial semialgebraic descriptions
of $3$-leaf model spaces, and develop several examples.  In \S \ref{sec:nleaf}, we
give the main result:  a semialgebraic description of the  $k$-state
general Markov model on $n$-leaf trees for nonsingular parameters.
For the $2$-state model
we prove a slightly stronger result that drops the nonsingularity assumption.

\section{Definitions and Notations}\label{sec:defns}

\subsection{The general Markov model on trees} We review the  
$k$-state general Markov model on trees, \GMk, whose parameters consist 
of a combinatorial object, a tree, and a collection of numerical parameters 
that are associated to a rooted version of the tree.

Let $T=(V,E)$ be a binary tree with leaves $L\subseteq V$, $|L|=n$,
and $\{X_a\}_{a\in V}$ a collection of discrete random variables
associated to the nodes, all with state space $[k]=\{1,2,\dots,k\}$.
Distinguish an internal node $r$ of $T$ to serve as its root, and
direct all edges of $T$ away from $r$.  (Often this
model is presented with the root as a node of valence 2 which is introduced by subdividing
some edge. However, under very mild assumptions this leads to the same
probability distributions we consider here \cite{SSH94, AR03}, so we
avoid that complication.)  Though necessary for parameterizing the
model, the choice of $r$ will not matter in our final results, as will be shown in \S \ref{sec:nleaf}.

For a tree $T$ rooted at $r$, numerical parameters $\{\boldsymbol \pi, \{M_e\}_{e\in E}\}$
for the \GMk model on $T$ are:
\begin{romannum}
\item A \emph{root distribution} row vector $\boldsymbol \pi 
= (\pi_1, \ldots, \pi_k)$, with nonnegative entries summing to 1;
\item  \emph{Markov matrices $M_e$}, with nonnegative entries 
and row sums equal to $1$.
\end{romannum}

The vector $\boldsymbol \pi$ specifies the distribution of the random
variable $X_r$, \emph{i.e.}, $\pi_i=\prob(X_r=i)$, and the Markov
matrices $M_e$, for $e = (a_e, b_e)\in E$, give transition
probabilities $M_e(i,j) = \prob ( X_{b_e} = j \mid X_{a_e} = i)$ of
the various state changes in passing from the parent vertex $a_e$ to 
the child vertex $b_e$.  Letting $\mathbf X=(X_a)_{a\in V}$, the joint probability
distribution at all nodes of $T$ is thus
$$\prob(\mathbf X=\mathbf j)=\pi_{j_r} \prod_{e\in E}M_e(j_{a_e},j_{b_e}).$$
By marginalizing over all variables at internal nodes of $T$, we obtain the \emph{joint
distribution, $P$, of states at the leaves of $T$}; if $\mathbf k\in [k]^{|L|}$ is
an assignment of states to leaf variables, then 
$$
P(\mathbf k)=\sum_{\mathbf m \in [k]^{|V\smallsetminus L|}} \prob( \mathbf X=(\mathbf k,\mathbf m))
$$
where $(\mathbf k, \mathbf m)$ is an assignment of states to all
the vertices of $T$ compatible with $\mathbf k$.
It is natural to view $P$ as an $n$-dimensional $k\times \cdots \times
k$ array, or \emph{tensor}, with one index for each leaf of the tree. 

For fixed $T$ and choice of $r$, we use
$\psi_T$ to denote the \emph{parameterization map}
$$\psi_T: \{\boldsymbol \pi, \{M_e\}_{e\in E}\} \mapsto P.$$
That the coordinate functions of $\psi_T$ are polynomial is obvious,
but essential to our work here.  Note that we may naturally extend
the domain of the polynomial map to larger sets, by dropping the
nonnegativity assumptions in (i) and (ii), but retaining the condition that rows must 
sum to 1.  We will consider
\emph{real parameters} and a real parameterization map, as well as
\emph{complex parameters} and a complex parameterization map. In
contrast, we refer to the original probabilistic model as having
\emph{stochastic parameters}. Since the parameterization maps are all
given by the same formula, we use $\psi_T$ to denote them all, but
will always indicate the current domain of interest.

The image of complex, real, or stochastic parameters under $\psi_T$
is an $n$-dimensional $k \times\cdots \times k$ tensor,
whose $k^n$ entries sum to 1.  When parameters are not stochastic,
this tensor generally does not specify a probability distribution, as
there can be negative or complex entries. We refer to any tensor whose
entries sum to 1, regardless of whether the entries are complex, real,
or nonnegative, as a \emph{distribution}, but reserve the term
\emph{probability distribution} for a nonnegative distribution.  With
this language, the image of complex parameters under $\psi_T$ is a
distribution, but may or may not be a probability distribution.
Similarly, while the matrix parameters $M_e$ have rows summing to one
even for complex parameters, we reserve the term \emph{Markov} matrix
exclusively for the stochastic setting.

\subsection{Algebraic and semialgebraic model descriptions} 

Most previous algebraic analysis of the \GMk model has focused on the
\emph{algebraic variety} associated to it for each choice of tree $T$.  
With this viewpoint one is
essentially passing from the parameterization of the model, as given above, to an
implicit description of the image of the parameterization as a zero set of certain polynomial
functions, traditionally called \emph{phylogenetic invariants} 
\cite{CF87, Lake87, ARgascuelchapter}.

Whether one considers stochastic, real, or complex parameters, the
collection of phylogenetic invariants for \GMk on a tree $T$ are the
same.  Thus they cannot distinguish probability
distributions that arise from stochastic parameters from those arising
from non-stochastic real or complex ones. To complicate matters further, 
there exist distributions that satisfy all
phylogenetic invariants for the model on a given tree, but are not even in
the image of complex parameters. Though the algebraic
issues behind this are well understood, they prevent classical algebraic
geometry from being a sufficient tool to focus exclusively on the distributions of statistical
interest.

To gain a more detailed understanding, we seek to refine the algebraic
description of the model given by phylogenetic invariants into a
\emph{semialgebraic} description: In addition to finding polynomials
vanishing on the image of the parameterization (or equivalently
polynomial equalities holding at all points on the image), we also
seek polynomial inequalities sufficient to distinguish the
stochastic image precisely.  

\smallskip

Recall that a subset of $\R^n$ is called a \emph{semialgebraic set} 
if it is a boolean combination of sets each of which is
defined by a single polynomial equality or inequality.
The Tarski-Seidenberg Theorem \cite{Tarski1951, Seidenberg1954} 
implies that the image of a 
semialgebraic set under a polynomial map is also semialgebraic.

Since for all $T$ the stochastic parameter space of $\psi_T$ is clearly 
semialgebraic, this implies that semialgebraic descriptions exist for  the
images of the $\psi_T$. Determining such descriptions explicitly 
is our goal. 

\subsection{Nonsingular parameters, positivity, and independence}\label{ssec:nonsing}

Some of our results will be stated with additional mild conditions placed on the 
allowed parameters for the \GMk model.
We state these conditions here, and explore their meaning.

\begin{definition}\label{def:nonsingparams} A choice $\{\boldsymbol \pi, \{M_e\}_{e\in E}\}$ of  stochastic, real, or complex 
parameters for \GMk on a tree $T$ with root $r$ is said to be \emph{nonsingular} provided
\begin{romannum}
\item at every (hidden or observed) node $a$, the marginal distribution $\mathbf
  v_a$ of $X_a$ has no zero entry, and
\item for every edge $e$, the matrix  $M_e$ is nonsingular.
\end{romannum}
Parameters which are not nonsingular are said to be \emph{singular}.
\end{definition}

For stochastic parameters, the first condition in this definition can be replaced
with a simpler one:
\begin{romannum}
\item[({\rm i}')] \emph{the root distribution $\boldsymbol \pi$ has no zero entry.  }
\end{romannum}
Statement (i) follows from (i') and (ii) inductively,
since if all entries of $\mathbf v_a$ are positive and $M_{(a,b)}$ is a nonsingular
Markov matrix, then the distribution $\mathbf v_b=\mathbf v_a M_{(a,b)}$ at a child $b$ of $a$ 
has positive entries. However, for complex or real 
parameters requirement (i) is not implied by (i') and (ii), as a simple example shows: $\mathbf
v_a=(1/2,1/2)$, and $M_{(a,b)}=\begin{pmatrix} s&1-s\\2-s&s-1\end{pmatrix}$ are
singular parameters since $\mathbf v_b=(1,0)$, even though $\mathbf v_a$ has no
zero entries and $M_{(a,b)}$ is a nonsingular for $s\ne 1$.

\medskip

It is also natural to require that all numerical parameters of \GMk on
a tree $T$ be strictly positive. This means that all states may occur
at the root, and every state change is possible in passing along any
edge of the tree.  This assumption is plausible from a modeling point
of view, and can be desirable for technical statistical issues as
well. Note that positivity of parameters does not ensure
nonsingularity, since a Markov matrix may be singular despite all its
entries being greater than zero. Similarly, nonsingularity of
parameters does not ensure positivity since a nonsingular Markov
matrix may have zero entries.

\medskip

Given a joint probability distribution of random variables, two subsets of variables are \emph{independent} when the marginal
distribution for the union of the sets is the product of the 
marginal distributions for the two sets individually. We also use this term, in a nonstandard way, to apply  
to complex or real distributions when the same factorization holds.

To illustrate this usage, consider a tree $T$ with two nodes, $r$, $a$
and one edge $(r,a)$. For complex parameters $\boldsymbol
\pi$ and $M_{(r,a)}$, the joint distribution of $X_r$ and  $X_a$ is
given by the matrix
$$P=\diag(\boldsymbol \pi)M_{(r,a)}.$$ 
Then the variables  are independent exactly when
$P$ is a rank $1$ matrix: $P=\boldsymbol \pi ^T\mathbf v_a$. For
$k=2$ this occurs precisely when the parameters are singular. For
$k>2$, however, independence implies the parameters are singular, but not \emph{vice versa}.
In general, 
singular parameters ensure that $P$
has rank strictly less than $k$,  but not that $P$ has rank $1$.

These comments easily extend to larger trees to give the following.

\smallskip

\begin{proposition}\label{prop:ind}
  Suppose $P=\psi_T(\boldsymbol \pi, \{M_e\})$ for a choice of complex
  GM($k$) parameters on an $n$-leaf tree $T$. If the parameters are
  nonsingular, then there is no proper partition of the indices of $P$
  into independent sets. For $k=2$, the converse also holds.
\end{proposition}

That the converse is false for $k>2$ is a  complicating factor
for the generalization of our results from the $k=2$ case.  Indeed, this
is the reason we ultimately restrict to nonsingular parameters.

\medskip

In closing this section, we note that for any $P \in \image (\psi_T)$,
there is an inherent and well-understood source of
non-uniqueness of parameters giving rise to $P$, sometimes called `label-swapping.'  
Since internal nodes of $T$ are unobservable variables, the distribution
$P$ is computed by summing over all assignments of states to such
variables. As a result, if the state names were permuted for such a variable, 
and corresponding changes made in numerical
parameters, $P$ is left unchanged. Thus parameters leading 
to $P$ can be determined at most up to such permutations.

In the case of nonsingular parameters, label-swapping is the only 
source of non-uniqueness of parameters leading to $P$ 
\cite{Chang96}. (See also \cite{Kruskal77,ARM09}). For singular 
parameters there are additional sources of non-uniqueness.

\subsection{Marginalizations, slices,  group actions, and flattenings}\label{ssec:marg}

Viewing probability distributions on $n$ variables as 
$n$-dimensional tensors gives natural associations between statistical
notions and tensor operations.  For example,
summing tensor entries over an index, or a collection 
of indices, corresponds to marginalizing over a
variable, or collection of variables. Considering only those
entries with a fixed value of an index, or collection of indices,
corresponds (after renormalization) to conditioning on an observed
variable, or collection of variables. Rearranging array entries into a
new array, with fewer dimensions but larger size, corresponds to
agglomerating several variables into a composite one
with larger state space. 
Here we introduce the necessary notation to formalize these tensor operations.

\begin{definition}  For an $n$-dimensional $k \times \cdots \times k$ tensor
$P$, integer $i\in[n]$, and vector $\mathbf v=(v_1, \cdots, v_k)$,  define the $(n-1)$-dimensional tensor
$P\ast_i \mathbf v$ by
$$(P*_i\mathbf v)(j_1,\dots,\hat j_i,\dots,j_n) = \sum_{j_i=1}^k v_{j_i} P(j_1,\cdots, j_i, \cdots,j_n),$$
where $\hat \ $ denotes omission.

Thus, 
the $\ell$th \emph{slice of $P$ in the $i$th index} is defined by
$P_{\cdots \ell \cdots}=P*_i \mathbf e_\ell,$ where $\mathbf e_\ell$ is the $\ell$th standard basis vector,
and the $i$th \emph{marginalization of $P$} is
$P_{\cdots+\cdots}=P*_i\mathbf 1$
where  $\mathbf 1$ is the vector of all $1$s.
\end{definition}

\smallskip

The
above product of a tensor and vector extends naturally to tensors and matrices.

\smallskip

\begin{definition}
For an $n$-dimensional $k \times \cdots \times k$ tensor $P$ and $k \times k$ matrix $M$, define the
$n$-dimensional tensor
 $P*_iM$  by
$$
(P \ast_i M) (j_1,  \dots, j_n) = \sum_{\ell=1}^k  P(j_1, \dots, j_{i-1}, \ell,  j_{i+1}, \dots, j_n) M(\ell,j_i).$$
\end{definition}

If the above operations on a tensor by vectors or
matrices are performed in different indices, then they commute. This allows the use of
$n$-tuple notation for the operation of matrices in all indices of a tensor, such as the following:
$$P \cdot (M_1,M_2,\dots, M_n)=(\cdots((P*_1M_1)*_2M_2)\cdots )*_nM_n.$$
Although the $M_i$ need not be invertible, restricting to that case gives
the natural (right) group action of $GL(k,\C)^n$ on 
$k\times\cdots \times k$ tensors.
This generalizes the familiar operation on 2-dimensional tensors $P$, \emph{i.e.}, on
matrices, where
$$P \cdot (M_1, M_2) = (P \ast_1 M_1) \ast_2 M_2=M_1^TPM_2.$$  
 
 \medskip
 
 If $\mathbf v\in \C^k$, then $\Diag(\mathbf v)$
 denotes the $3$-dimensional $k \times k \times k$ diagonal
 tensor whose only nonzero entries are 
the $ v_i$ in the $(i,i,i)$ positions. That this notion is useful for the \GMk model is made clear
 by the observation that for a 3-leaf star tree $T$, rooted at the central node,
\begin{equation}\psi_T(\boldsymbol \pi,\{M_1,M_2,M_3\})=\Diag(\boldsymbol \pi) \cdot (M_1,M_2,M_3).\label{eq:3leafdist}
\end{equation}
\medskip 

If $P$ is an $n$-dimensional $k \times \cdots \times k$
tensor and $[n]=A\sqcup B$ is a disjoint union of nonempty sets, then
the \emph{flattening of $P$ with respect to this bipartition},
$\Flat_{A|B}(P)$ is the $k^{|A|}\times k^{|B|}$ matrix with rows
indexed by $\mathbf i\in [k]^{|A|}$ and columns indexed by $\mathbf
j\in [k]^{|B|}$, with
$$\Flat_{A|B}(P) (\mathbf i,\mathbf j)=P(\mathbf k),$$
where $\mathbf k\in [k]^n$ has entries matching those of $\mathbf i$
and $\mathbf j$, appropriately ordered. Thus the entries of $P$ are
simply rearranged into a matrix, in a manner consistent with the
original tensor structure. When $P$ specifies a joint distribution for
$n$ random variables, this flattening corresponds to treating the
variables in $A$ and $B$ as two agglomerate variables, with state
spaces the product of the state spaces of the individual variables.

Notations such as $\Flat_{1 \vert 23} (P)$, for example, will be used
to denote the matrix flattening obtained from a $3$-dimensional tensor
using the partition of indices $A=\{1\}$, $B=\{2,3\}$.  If $e$ is an
edge in an $n$-leaf tree, then $e$ naturally induces a bipartition of
the leaves, by removing the edge and grouping leaves according to the resulting connected components.  A flattening
for such a bipartition is denoted by $\Flat_e (P)$.

Finally, we note that flattenings naturally occur in the notion of independence:
If $[n]=A\sqcup B$, then the sets are independent precisely when
$\Flat_{A|B}(P)$ is a rank 1 matrix.

\section{GM($k$) on  $3$-leaf trees}\label{sec:THREEleaf}

In this section we derive a semialgebraic
description of GM($k$) on the $3$-leaf tree, the smallest example of interest.
Results for the 3-leaf tree also
serve as a building block for the study of
the model on larger trees in \S \ref{sec:nleaf}.
For this section, then, $T$ is fixed, with leaves $1$, $2$, $3$ and
root $r$ at the central node.

When $k=2$, Cayley's \emph{hyperdeterminant} plays a critical
role, as has already been highlighted in \cite{SZ2012}.  
Though our formulation will be different, we 
take
the hyperdeterminant as our starting point.
For any $2 \times 2 \times 2$ tensor $A=(a_{ijk})$,  
the hyperdeterminant $\Delta (A)$ \cite{Cayley, GKZ,dSL} is
\begin{multline*}
\Delta (A) =
\big(a_{111}^2 a_{222}^2+a_{112}^2 a_{221}^2+a_{121}^2  a_{212}^2+a_{122}^2 a_{211}^2\big)
-2\big(a_{111} a_{112} a_{221} a_{222}\\
+a_{111} a_{121} a_{212} a_{222}+a_{111} a_{122} a_{211} a_{222} +a_{112} a_{121} a_{212} a_{221}+a_{112} a_{122} a_{221} a_{211}\\
+a_{121} a_{122} a_{212} a_{211}\big)+4 \big(a_{111} a_{122} a_{212} a_{221}+a_{112} a_{121} a_{211} a_{222}\big).
\end{multline*}

The function $\Delta$ has the $GL(2,\C)^3$-invariance property
\begin{equation}\Delta(P \cdot (g_1,g_2,g_3))=\det(g_1)^2\det(g_2)^2 \det(g_3)^2 \Delta(P).\label{eq:Deltainvariance}\end{equation}
This fact, combined with a study of canonical forms for $GL(2,\C)^3$-orbit
representatives, leads to the following theorem.  

\begin{theorem}[\cite{dSL}, Theorem 7.1]\label{thm:dSL}
A complex $2\times 2\times 2$ tensor 
$P$ is in the $\GL (2, \C)^3$-orbit of $D = \Diag(1,1)$ if, and only if, 
$\Delta (P) \ne 0$.
A real tensor is in the $\GL (2, \R)^3$-orbit of $D$ if, and only if, 
$\Delta (P) > 0$.
\end{theorem}

\medskip

Suppose that $k=2$ and $P=\psi_T(\boldsymbol \pi, \{M_1,M_2,M_3\})$
arises from nonsingular parameters on $T$. Then, equation
\eqref{eq:3leafdist} states $P = \Diag(\boldsymbol \pi) \cdot (M_1,
M_2, M_3),$ but letting $M_1'=\diag(\boldsymbol \pi)M_1$ we also
have $$P=D \cdot (M_1', M_2, M_3).$$ Therefore $\Delta(P)>0$ by the
forward implication of Theorem~\ref{thm:dSL}. This hyperdeterminantal
inequality can thus be included in building a semialgebraic
description of the GM($2$) model when restricted to nonsingular
parameters.

However, the inequality $\Delta(P)>0$ yields a weaker conclusion than
that $P$ arises from stochastic, or even real, nonsingular parameters, so
additional inequalities are needed for a semialgebraic model
description.

\smallskip

Nonetheless, motivated by the role the hyperdeterminant plays in the
semialgebraic description of the GM($2$) model, in a separate work
Allman, Jarvis, Rhodes, and Sumner \cite{AJRS2012} construct
generalizations of $\Delta$ for $k\ge2$.  These functions are defined
by
$$f_i(P;\mathbf x)=\det(H_\mathbf x( \det(P*_i\mathbf x))),$$
where $\mathbf x$ is a vector of auxiliary variables, and $H_\mathbf x$ denotes the Hessian operator. They also
have invariance properties under $GL(k,\C)^3$ such as
$$f_3(P\cdot(g_1,g_2,g_3);\mathbf x)=\det(g_1)^k\det(g_2)^k\det(g_3)^{2}f_3(P;g_3\mathbf x).$$

The next theorem establishes that the nonvanishing of these polynomials, in 
conjunction with the vanishing of some others, identifies the orbit of $\Diag(\mathbf 1)$, and thus is an analog of
Theorem \ref{thm:dSL} for larger $k$.

\smallskip

\begin{theorem}[\cite{AJRS2012}]\label{thm:AJRS} 
  A complex $k\times k \times k$ tensor $P$ lies in
  the $GL(k,\C)^3$-orbit of $\Diag(\mathbf 1)$ if, and only if, for
  some $i\in\{1,2,3\}$,
\begin{romannum}
\item \label{cond:inv} $(P*_i\mathbf e_j)\adj (P*_i\mathbf x)
  (P*_i\mathbf e_\ell)-(P*_i\mathbf e_\ell)\adj (P*_i\mathbf x) (P*_i\mathbf
  e_j) = 0$ for all $j,\ell\in[k]$. Here $\adj$
 denotes the classical adjoint, and equality means
as a matrix of polynomials in $\mathbf x$.
\item $f_i(P;\mathbf x)$ is not identically zero as a polynomial in $\mathbf x$.
\end{romannum}
Moreover, if the enumerated  conditions hold for one $i$, then they hold for all.
\end{theorem}
\smallskip

When $k>2$ the $GL(k,\C)^3$-orbit of $\Diag(\mathbf 1)$ is not dense among
all $k\times k\times k$ tensors; rather its closure is a lower dimensional
subvariety.  This explains the necessity of the
equalities in item (i).
In the case $k=2$ these equalities simplify to $0=0$ and thus hold for all tensors. 
One can further verify that if $k=2$ then $f_i=\Delta$, so that Theorem \ref{thm:AJRS} includes 
 the first statement of Theorem \ref{thm:dSL}.

One might hope that the polynomials $f_i(P;
\mathbf x)$ had a sign property similar to that given in Theorem \ref{thm:dSL} for $\hyper$, so that a simple
test could further distinguish the image of nonsingular real
parameters. For $k=3$, using functions related to the $f_i$, a 
 semialgebraic description of the $GL(k,\R)^3$-orbit of $\Diag(\mathbf
 1)$ can in fact be obtained in this manner (see  \cite{AJRS2012}), giving a complete analog of Theorem \ref{thm:dSL}. However, for $k>3$
 no analog is known.

Finally, we emphasize that for $k>2$ the functions $f_i$ are \emph{not}
 the ones usually referred to as hyperdeterminants \cite{GKZ}, but
 rather a different generalization of $\Delta$. 
 
\medskip

With semialgebraic conditions ensuring a tensor is in the
$GL(k,\C)^3$ orbit of $\Diag(\mathbf 1)$ in hand, we wish to supplement these
to ensure it arises from nonsingular stochastic parameters.  We address this in several steps; first,
we give requirements that a tensor is the image of complex parameters
under $\psi_T$, and then that these parameters be  nonnegative.

\begin{proposition}\label{prop:complexkstates}
Let $P$ be a complex $k \times k \times k$  distribution. Then
$P$ is in the image of nonsingular complex parameters for GM($k$) on
the $3$-leaf tree if, and only if, $P$ is in the
$GL(k,\C)^3$-orbit of $\Diag(\mathbf 1)$ and
$\det (P \ast_i \ones) \neq 0$ for $i=1,2,3$.
Moreover, the parameters are unique up to label swapping. 
\end{proposition}

\begin{proof} To establish the claimed reverse implication, suppose $P= \Diag(\mathbf 1) \cdot (g_1, g_2, g_3)$ for some
  $g_i \in \GL (k ,\C)$, and let $\mathbf r^i=g_i\mathbf 1$ denote the
  vector of row sums of $g_i$. A computation shows that
$$
P \ast_3 \mathbf 1 = g_1^T \diag(\mathbf r^3) g_2.
$$
Thus $\det (P \ast _3 \mathbf 1) \neq 0$ is equivalent to the row sums of 
$g_3$ being nonzero, and similarly for the other $g_i$.

Now $M_i =
{\diag(\mathbf r^i)}^{-1} g_i$ is a complex matrix with row sums equal to
one.  Letting $\boldsymbol \pi = ( \prod_{i=1}^3 r^i_1, \dots,
\prod_{i=1}^3 r^i_k)$ be the vector of entry-wise products of the
$\mathbf r^i$, the entries of $\boldsymbol \pi$ are nonzero and $$P = \Diag(\boldsymbol
\pi) \cdot (M_1, M_2, M_3).$$ Since $P$ is a distribution,
\begin{align*}1&=((P*_1\mathbf 1)*_2\mathbf 1)*_3\mathbf 1\\
&=(((\Diag(\boldsymbol \pi) \cdot (M_1, M_2, M_3))  \ast_1 \mathbf 1) \ast_2 \mathbf 1)\ast_3 \mathbf 1\\
&= ((\Diag(\boldsymbol \pi)  \ast_1 M_1 \mathbf 1) \ast_2 M_2  \mathbf 1) \ast_3 M_3  \mathbf 1 \\
&= ((\Diag(\boldsymbol \pi)  \ast_1  \mathbf 1) \ast_2  \mathbf 1) \ast_3 \mathbf 1 \\
&=\boldsymbol \pi\cdot \mathbf 1,
\end{align*} so $\boldsymbol \pi$ is a valid complex root distribution.
Thus, $P$ is in the image of $\psi_T$ for complex, nonsingular parameters.

The forward implication in the theorem is straightforward.

The uniqueness of  nonsingular parameters up to permutation of states
at the internal node of the tree was discussed
at the end of subsection \ref{ssec:nonsing}.
 \qquad\end{proof}

\smallskip

Combining this Proposition with Theorems \ref{thm:dSL} and \ref{thm:AJRS} we
obtain the following.

\begin{corollary}\label{cor:complexToMarkov}
A $k \times k \times k$ complex distribution $P$
is the image of complex, nonsingular parameters for GM($k$)
on the $3$-leaf tree if, and only if, it satisfies the semialgebraic conditions 
(i) and (ii) of Theorem \ref{thm:AJRS} and
\begin{remunerate}
\item[(iii)] for $i = 1, 2, 3$, $\det (P \ast_i \mathbf 1) \neq 0$.
\end{remunerate}

For $k=2$,  $P$ is the image of real nonsingular parameters for GM($2$)
on the $3$-leaf tree if, and only if, it
 satisfies $\Delta(P) > 0$ and the semialgebraic conditions (iii).
\end{corollary}

\smallskip 

Next we characterize the image of nonsingular stochastic
parameters, and finally of strictly positive nonsingular parameters. 
The key to this step is the construction of certain quadratic forms whose 
positive semi-definiteness (respectively definiteness)
encodes nonnegativity (respectively positivity) of some of the numerical 
parameters. Sylvester's 
Theorem \cite{Sylvester1852}, which we state for reference here, then gives 
a semialgebraic version of these conditions.

Recall that a \emph{principal minor} of a matrix is the determinant of
a submatrix chosen with the same row and column indices, and that a
\emph{leading} principal minor is one of these where the chosen
indices are $\{1,2,3,\dots,k\}$ for any $k$.

\begin{theorem}[Sylvester's Theorem]\label{thm:sylvester}
Let $A$ be an $n\times n$ real symmetric matrix and 
$Q(\mathbf x)=\mathbf x^TA\mathbf x$ the associated quadratic form on $\R^n$.
Then
\begin{remunerate} 
\item $Q$ is positive semidefinite if, and only if, all principal minors of $A$
  are nonnegative, and
\item $Q$ is positive definite if, and only if, all leading principal minors of
  $A$ are strictly positive.
\end{remunerate}
\end{theorem}

\medskip

We use Sylvester's Theorem to establish the following theorem.

\smallskip

\begin{theorem}\label{thm:kstate-quadform}
A  $ k \times k \times k$ tensor $P$  is the image of nonsingular stochastic parameters 
for the GM($k$) model on the $3$-leaf tree if, and only
if, its entries are nonnegative and sum to $1$, conditions (i), (ii), and (iii) of 
Theorem \ref{thm:AJRS} and Corollary \ref{cor:complexToMarkov} are satisfied, and
\begin{remunerate}
\item[(iv)]
all leading principal minors of 
\begin{equation}
\det(P_{\cdot\cdot+})P_{+\cdot\cdot}^T \adj(P_{\cdot\cdot +}) P_{\cdot +\cdot},\label{eq:quadform-matrices1}
\end{equation}
 are positive, and all principal minors of the following matrices are nonnegative:
\begin{eqnarray}
&\det(P_{\cdot\cdot+}) \, P_{i\cdot\cdot}^T \, &\adj(P_{\cdot\cdot+}) P_{\cdot + \cdot}, 
\quad \text{for } i=1, \dots, k \label{eq:quadform-matrices},\\
&\det(P_{\cdot\cdot+}) P_{+\cdot\cdot}^T &\adj (P_{\cdot\cdot+}) P_{\cdot i \cdot}, 
\quad  \text{ for } i=1, \dots, k, \nonumber  \\
&\det(P_{+\cdot\cdot})P_{\cdot+\cdot} &\adj(P_{+\cdot\cdot}) P_{\cdot \cdot i}^T, 
\quad \text{ for } i=1, \dots, k. \nonumber
\end{eqnarray}
\end{remunerate}

Moreover, $P$ is  the image of nonsingular positive parameters if,
and only if, its entries are positive and sum to $1$, conditions (i), (ii), and (iii) 
are satisfied and
\begin{remunerate}
\item[(iv')] all leading principal minors of the matrices in 
\eqref{eq:quadform-matrices1} and \eqref{eq:quadform-matrices} are positive.
\end{remunerate}

In both of these cases, the nonsingular parameters are unique up to label swapping.
\end{theorem}

\smallskip 

\begin{proof}Let $P$ be an arbitrary nonnegative $k\times k\times k$
  tensor whose entries sum to 1.  By Corollary
  \ref{cor:complexToMarkov}, the first 3 conditions are equivalent to $P=\psi_T(\boldsymbol \pi, \{M_1,M_2,M_3\})$
  for complex nonsingular parameters.  We need to show the addition of
  assumption (iv) is equivalent to parameters being
  nonnegative. 
  
    Note that
\begin{align*}
P_{\cdot\cdot+}&=P*_3\mathbf 1= M_1^T\diag(\boldsymbol \pi)M_2,\\
P_{\cdot+\cdot}&=P*_2\mathbf 1=M_1^T\diag(\boldsymbol \pi)M_3,\\ 
P_{+\cdot\cdot}&=P*_1\mathbf 1=M_2^T  \diag(\boldsymbol \pi)M_3.
\end{align*}
Since $P_{\cdot\cdot+}$
is nonsingular, we find 
 \begin{equation}P_{+\cdot\cdot}^TP_{\cdot\cdot+}^{-1}P_{\cdot+\cdot}=M_3^T\diag(\boldsymbol \pi) M_3\label{eq:symprod}\end{equation}
  is symmetric, and
  the matrix of a positive definite quadratic form if, and only if the entries of
  $\boldsymbol \pi$ are positive. 
  Equivalently, by Sylvester's theorem, all leading principal minors of this matrix
  must be positive.

  Similarly, using slices one has $$P_{i\cdot\cdot}^T
  P_{\cdot\cdot+}^{-1}P_{\cdot + \cdot}=M_3^T\diag(\boldsymbol
  \pi)\Lambda_{1,i} M_3$$ where $\Lambda_{1,i}=\diag(M_1\mathbf e_i)$
  is the diagonal matrix with entries from the $i$th column of
  $M_1$. Thus its principal minors being nonnegative is equivalent
  (since we already have that the entries of $\boldsymbol \pi$ are
  positive) by Sylvester's Theorem to the entries in the $i$th column of $M_1$ being
  nonnegative. The product $P_{+\cdot\cdot}^T
  P_{\cdot\cdot+}^{-1}P_{\cdot i \cdot}$ similarly can be used for a
  condition that the $i$th column of $M_2$ be nonnegative, and the
  product $P_{\cdot+\cdot} P_{+\cdot\cdot}^{-1}P_{\cdot \cdot i}^T$
  for the columns of $M_3$. 
  
  Multiplying all these matrices by the square of
  an appropriate nonzero determinant clears denominators and preserves
  signs, yielding \eqref{eq:quadform-matrices1} and \eqref{eq:quadform-matrices}. 

For the second statement, that the matrices in
  \eqref{eq:quadform-matrices1} and \eqref{eq:quadform-matrices}
 have positive leading principal minors is equivalent, by Sylvester's
  Theorem, to the positive definiteness of the quadratic forms, which in turn is equivalent to
  the positiveness of parameters.  Since these parameters are nonsingular, the only source
  of non-uniqueness is label swapping.
 \qquad\end{proof}
\smallskip

\emph{Remark.}
Matrix products such as that of equation \eqref{eq:symprod} appeared in 
\cite{AR03}, where their symmetry was used to produce phylogenetic 
invariants, but their usefulness for stating nonnegativity of parameters was overlooked.

\smallskip

\emph{Remark.} The $j\times j$ minors of the matrices in \eqref{eq:quadform-matrices1} 
and \eqref{eq:quadform-matrices} are
polynomials in the entries of $P$ of degree $j(2k+1)$, with
$j=1,\dots, k$. However, as the leading determinant in those products  is real and nonzero, 
one can remove an even power of it without affecting the sign of the minors. Thus 
the polynomial inequality of degree $j(2k+1)$ can be replaced by one of lower degree,
$j(k+1)+e_jk$, where $e_j=0$ or $1$ is the parity of $j$.

\smallskip

In the case of the $2$-state model, the above result can be made more
complete, by also explicitly describing the image of singular parameters.  
While semialgebraic characterizations of probability distributions for both nonsingular and
singular parameters on the 3-leaf tree have been given previously by
\cite{SettimiSmith2000}, \cite{Auvray2006}, \cite{KL2011}, and \cite{SZ2011}, 
we provide another since our approach
is novel.

\smallskip

\begin{theorem}\label{thm:nonneg}
  A tensor $P$ is in the image of the stochastic parameterization map $\psi_T$ for the
  GM(2) model on the 3-leaf tree if, and only if, its entries are nonnegative and 
  sum to 1, and one of the following occur:
\begin{remunerate}
\item \label{cond:2pos} $\Delta(P)>0$, $\det(P*_i\mathbf 1)\ne 0$ for $i=1,2,3$,
  all leading principal minors of
 $$ \det(P_{\cdot\cdot+})P_{+\cdot\cdot}^T \adj(P_{\cdot\cdot +}) P_{\cdot +\cdot}$$
 are positive, and all principal minors
 of the following six
 matrices are nonnegative:
\begin{align*}
\det(P_{\cdot\cdot+}) \, P_{i\cdot\cdot}^T \, &\adj(P_{\cdot\cdot+}) P_{\cdot + \cdot}, \quad \text{for } i=1,2,\\
\det(P_{\cdot\cdot+}) P_{+\cdot\cdot}^T &\adj (P_{\cdot\cdot+}) P_{\cdot i \cdot},\quad  \text{ for } i=1,2,\\
\det(P_{+\cdot\cdot})P_{\cdot+\cdot} &\adj(P_{+\cdot\cdot}) P_{\cdot \cdot i}^T, \quad \text{ for } i=1,2.
\end{align*}
In this case, $P$ is the image of unique (up to label swapping) nonsingular
parameters.
\item \label{cond:2zero} $\Delta(P)=0$, and all $2 \times 2$ minors of at least one of the matrices
  $\Flat_{1|23}(P)$, $\Flat_{2|13}(P)$, $\Flat_{3|12}(P)$ are zero. 
In this case, $P$ arises from singular parameters. If $P$ has all positive entries, then it is the image of infinitely many singular stochastic parameter choices.  
\end{remunerate}
\end{theorem}

\smallskip

\begin{proof} Using Theorem \ref{thm:kstate-quadform}, and the observations 
made for $k=2$ immediately following the statement of Theorem \ref{thm:AJRS},
case 1 is already established under the weaker condition that $\Delta(P)\ne 0$. 
However, since the parameters are nonsingular and real when $\Delta(P) \ne 0$
and the conditions of case 1 are satisfied,
by Theorem \ref{thm:dSL}
we may assume equivalently that $\Delta(P) > 0$.

\smallskip

To establish case 2, first assume $P=\psi_T(\boldsymbol \pi,\{M_1,M_2,M_3\})$ 
is the image of singular stochastic parameters.
Then certainly $P$ has nonnegative entries summing to 1, and  by equations 
\eqref{eq:3leafdist} and \eqref{eq:Deltainvariance}, $\Delta(P)= 0$. Since 
$$\Flat_{1|23}(P)=M_1^T\diag(\boldsymbol \pi) M,$$
where $M$ is the $2\times 4$ matrix obtained by taking the tensor product of 
corresponding rows of $M_2$ and $M_3$, this flattening has rank 1, if
$\boldsymbol \pi$ has a zero entry or $M_1$ has rank 1. Similar products 
for the other flattenings show that singular parameters imply at least one 
of the flattenings $\Flat_{1|23}(P), \Flat_{2|13}(P), \Flat_{3|12}(P)$ has rank 1, 
and hence its $2\times 2$ minors vanish.

Conversely, suppose $\Delta(P)=0$ and at least one of the flattenings has 
vanishing $2\times2$ minors, and hence rank 1. Then by the classification of orbits given in \cite[Table 7.1]{dSL}, 
$P$ is in the $GL(2,\R)^3$-orbit of one of the following four tensors:
the tensor $\Diag(1,0)$ (in which case 
all three flattenings have rank 1) or one of the 3 tensors with parallel 
slices $I$ and the zero matrix (in which case exactly one of the 
flattenings has rank 1).

If $P=\Diag(1,0)\cdot(g_1,g_2,g_3)$, then $P(i,j,k)=g_1(1,i)g_2(1,j)g_3(1,k)$. Since 
the entries of $P$ are nonnegative and sum to 1, one sees the top rows of 
each $g_i$ can be chosen to be nonnegative, summing to 1. The bottom 
row of each $g_i$ can also be replaced with any nonnegative row summing 
to 1 that is independent of the top row. Taking $\boldsymbol \pi=(1,0)$, this 
gives us infinitely many choices of singular stochastic parameters giving 
rise to $P$. Alternatively, one could  choose each Markov matrix to have 
two identical rows, and any $\boldsymbol \pi$ with nonzero entries to obtain 
other singular stochastic parameters leading to $P$.

For the remaining cases assume, without loss of generality, that $P=
E\cdot(g_1,g_2,g_3)$, where $E_{\cdot\cdot1}=(1/2) \, I$, and
$E_{\cdot\cdot2}$ is the zero matrix.  Then $P_{\cdot\cdot1}=g_3(1,1)
(g_1^Tg_2)$ and $P_{\cdot\cdot2}=g_3(1,2) (g_1^Tg_2)$.  Since the
entries of $P$ are nonnegative and add to 1, we may assume that the
top row of $g_3$ is also nonnegative and adds to 1. Choose $M_3$ to
have two identical rows matching the top row of $g_3$.  Now
$P_{\cdot\cdot+}=g_1^Tg_2$ is a rank-2 nonnegative matrix with entries
adding to 1. Such a matrix can be written in the form
$P_{\cdot\cdot+}=M_1^T\diag(\boldsymbol \pi)M_2$ with, for instance
$M_1=I$, $\boldsymbol \pi=P_{\cdot++}$, $M_2=\diag(\boldsymbol
\pi)^{-1}P_{\cdot\cdot+}$.  Then one has $P=\psi_T(\boldsymbol \pi,
\{M_1,M_2,M_3\}).$ If $P$ has positive entries one may also choose
$M_1$ sufficiently close to $I$ so that $M_1$, $\boldsymbol
\pi=(M_1^T)^{-1}P_{\cdot++}$, $M_2=\diag(\boldsymbol
\pi)^{-1}{{(M_1^T)}^{-1}}P_{\cdot\cdot+}$ all have nonnegative
entries, thus obtaining infinitely many singular parameter choices
leading to $P$. (The example of $P=E$ shows that with only nonnegative
entries there may be only finitely many singular parameter choices
leading to $P$.)  \qquad\end{proof}

\smallskip

\emph{Remark.} The analysis of the singular parameter case in this proof, by 
appealing without explanation to \cite[Table 7.1]{dSL}, has
not made explicit the importance of the notion of \emph{tensor rank.}
Indeed, that concept is central to both \cite{dSL} and \cite{AJRS2012}
and thus plays a crucial behind-the-scenes role in this work as well. 
The first singular case, a tensor in the orbit of $\Diag(1,0)$, is of tensor 
rank $1$, while the second,  a tensor in the orbit of $E$, is of tensor 
rank $2$ yet multlilinear rank $(2,2,1)$. The nonsingular case is those 
of tensor rank $2$ and multlinear rank $(2,2,2)$.

\emph{Remark.} In case \ref{cond:2pos} of the theorem, the polynomial
inequalities are of degree 4 and 2 (from $\Delta$ and the
determinants) and degree 5 and 10 (from the minors), but the degree 10
ones can be lowered to degree 6 by removing a factor of a determinant
squared.  The polynomial equalities appearing in case \ref{cond:2zero}
have degree $4$ and $2$, with the quadratics simply expressing that
one of the leaf variables is independent of the others.

\smallskip

Minor modifications to the argument give the extension to positive
parameters below.

\begin{theorem}\label{thm:pos}
  A tensor $P$ is in the image of the positive stochastic parameterization map for the GM(2)
  model on the 3-leaf tree if, and only if, its entries are positive and 
  sum to 1, and the conditions of Theorem \ref{thm:nonneg} are met 
with the following modification to case 1:  all leading principal minors
  of the matrices are positive.
\end{theorem}

\begin{proof}
  Case 1 is proved by combining the arguments for Theorems \ref{thm:kstate-quadform} and \ref{thm:nonneg}.
  
  For case 2, if $P$ is in the orbit of $\Diag(1,0)$ the argument of Theorem \ref{thm:nonneg}
  still applies, replacing `nonnegative' with 
  `positive', and using the second construction of singular parameters.  If $P$ is in the orbit of $E$
  we simply replace `nonnegative' with `positive' in the argument.
 \qquad\end{proof}

\smallskip

As mentioned above, semialgebraic descriptions of the binary general Markov 
model on trees have been given previously, but in ways 
where generalizations to $k$-state models were not apparent. Although 
we have considered only the $3$-leaf tree (with larger trees to be discussed 
in \S \ref{sec:nleaf}) thus far, we pause here to discuss some connections to the 
recent works of Zwiernik and Smith \cite{SZ2011} 
and Klaere and Liebscher \cite{KL2011} on semialgebraic 
descriptions of the binary model on trees, as well as the older work of 
Pearl and Tarsi \cite{PT86}.  

Though using different approaches (and in the case of \cite{PT86} with different goals), 
all three of these works
emphasize statistical interpretations of various quantities computed from a 
probability distribution $P$ (\emph{e.g.}, covariances, conditional covariances, 
moments, tree cumulants).  While analogs of some of the same quantities 
appear in our generalization to 
$k$-state models, we have used algebra, rather than statistics, to guide our 
derivation. Although an inequality such as $\det(P*_i\mathbf 1)\ne 0$ which appears 
in our description can be given a simple statistical interpretation when $k=2$ 
(that two leaf variables are not independent), for larger $k$ its meaning is more 
subtle, as it is tied to our notion of nonsingular parameters. Thus our generalization 
to $k$-state models uses a more detailed development than the simple 
generalization of statistical concepts from $k=2$ to larger $k$.

The role played by the hyperdeterminant $\Delta$ in giving a 
semialgebraic model description for $k=2$ was first made clear in \cite{SZ2011}.
Its role was essentially independently discovered in \cite[Theorem 6]{KL2011}, 
though without recognition that it is a classical algebraic object. Indeed, both 
works recognize $\Delta$ as an expression in the 2-variable and 3-variable 
covariances (\emph{i.e.}, central moments). This is a fascinating 
intertwining of algebra and statistics, yet we did not find it helpful in understanding 
the correct analog of $\Delta$ for higher $k$; rather, 
\cite{AJRS2012} develops the analog $f_i$ used here through algebraic 
motivation entirely. It would nonetheless be interesting to understand whether $f_i$ can be 
described in more statistical language.

The explicit semialgebraic model descriptions for the $2$-state 
model given in both \cite{SZ2011} and \cite{KL2011} take on quite
different forms than ours. This is not a surprise as such a description 
is far from unique, and different reasoning may produce different 
inequalities. The version in \cite{SZ2011}, for example, is stated in a 
different coordinate system, using tree cumulants rather 
than the entries of $P$. We find all these descriptions 
valuable, as what should be considered the simplest, or most natural, description is not obvious.

The focus of \cite{PT86} is on
recovery of  parameters for \GMtwo on the $3$-leaf tree from a probability distribution 
assumed to have arisen
from stochastic parameters, in an approach based on earlier work 
in latent structure analysis \cite{LH68}.  In addressing this question, however,
semialgebraic conditions on the distribution are obtained.  For instance, the non-vanishing of denominators is needed for formulas 
to make sense, and thus certain polynomials must be nonzero. (The authors 
seem to assume nonsingularity of parameters, though that is never clarified in the paper.)
While $\Delta$ never arises in \cite{PT86},
it is remarkable to note, then, that using the results of Theorem \ref{thm:dSL}
and making explicit the tacit assumptions that various rational expressions exist and
are real, the conditions given in Theorem 1 of \cite{PT86} are sufficient
to show $\hyper > 0$, and thus $P$ is in the image of stochastic nonsingular parameters. Thus
one can extract a semialgebraic model description from this work, even if that was not its goal.

\section{Inequalities for 3-leaf trees via Sturm theory}\label{sec:Sturm}

The semialgebraic model descriptions given in the previous section
have the advantage of being easily describable in a uniform way for
all $k$. However, semialgebraic descriptions are not unique, and there
is no clear notion of what description should be considered simplest.
It is also of interest, therefore, to obtain alternative polynomial
inequalities, possibly of lower degree, that must also be satisfied by
probability distributions in the model, in the hopes that they lead to
another, perhaps better, semialgebraic model description, or that they
might be of further use for testing whether a distribution arises from
the model. We explore another approach to doing so here.

\subsection{Review of Sturm Theory}

Sturm theory can be used to impose conditions that roots of a
univariate polynomial lie in a certain interval. We briefly recall
basic definitions and results.  Suppose that $f(x) \in \R[x]$ is a
non-constant polynomial of degree $m$, with no multiple roots, and we
wish to count the roots of $f$ in an interval $(a,b)$ where
$f(a)f(b)\ne 0$.  Then a \emph{Sturm sequence} $\mathcal S$ for $f$ on
the interval $[a,b]$ is a sequence $f = f_0, f_1, \dots, f_m$ of
polynomials satisfying certain sign relationships at the zeros of the
$f_j$ in the interval $[a,b]$.  We give an example 
below, but for specifics, see, for instance, \cite{Jacobson}.  
For any $c \in [a,b]$ which is not a root of any $f_i$, the \emph{sign variation}
$V_S(c)$ is the number of sign changes in the sequence $f(c), f_1(c),
\ldots, f_m(c)$.

\begin{theorem}[Sturm's Theorem] \label{thm:Sturms-thm}
Let $f(x) \in \R[x]$ be a non-constant polynomial and $\mathcal S$
a Sturm sequence for $f$ on $[a,b]$.  Then 
the number of distinct roots of $f(x)$ in $(a,b)$ is equal to
$V_{\mathcal S} (a)-V_\mathcal{S} (b)$.  
\end{theorem}

\smallskip

Though other constructions of Sturm sequences exist, we use the
\emph{standard sequence}, derived using a modified Euclidean
algorithm.  If $f = f_0$ is a polynomial of degree $m > 0$, then set
$f_1 = f'$, and for $j=2, \ldots, m$, take $f_j$ to
be the opposite of the 
remainder of division of $f_{j-2}$ by $f_{j-1}$.  This yields a Sturm
sequence for $f$ on any interval $[a,b]$ with $f_j(a) \neq 0$, $f_j(b)
\neq 0$.

To illustrate, suppose that $f(x) = x^2 + c_1 x +c_0 \in \R[x]$.  Then
the standard sequence is $f, 2x+c_1, \frac{c_1^2}{4} -c_0$.  For the
particular choice of coefficients $c_1 = -\frac{3}{4}$ and $c_0 =
\frac{1}{8}$, using this sequence on $[0,1]$, we calculate that
$f_0(0), f_1(0), f_2(0)$ is the sequence $\frac{1}{8}, -\frac{3}{4},
\frac{1}{64}$, and $f_0(1), f_1(1), f_2(1) = \frac{3}{8}, \frac{5}{4},
\frac{1}{64}$.  Thus, $V_{\mathcal S}(0) = 2$, $V_{\mathcal S} (1) =
0$, so $f$ has $V_{\mathcal S}(0) - V_{\mathcal S}(1) = 2$ roots in
$(0,1)$. Indeed, the factorization $f(x) =
\big(x-\frac{1}{2}\big)\big(x-\frac{1}{4}\big)$ shows this directly.

Two comments are in order: First, in this example observe
that $f_2(x)$ is one-fourth the discriminant of $f(x)$, and its
positivity ensures that a monic quadratic has distinct and real
roots. This shows that Sturm theory can produce familiar
algebraic expressions, such as the quadratic discriminant, and thus gives a tool for generalizing them.  The second observation we state more formally,
as it is needed in our arguments below.

\begin{corollary}\label{cor:signs}
If $f(x)$ is of degree $m$ with neither $0$ nor $1$ a root, and $S$ is a Sturm 
sequence for $f$ on $[0,1]$,  then $f$ has $m$
distinct roots in this interval precisely when $V_\mathcal S(0) =m$ 
and $V_\mathcal S(1) = 0$.
\end{corollary}

This corollary allows us to obtain inequalities ensuring a polynomial has distinct roots in $[0,1]$: One simply requires that
the $f_i(0)$ alternate in being $>0$ or $<0$, while the $f_i(1)$ either be all $>0$ or all $<0$. We informally refer to inequalities
obtained in this manner as \emph{Sturm sequence inequalities}.

\subsection{Eigenvalues and Sturm Sequences} 

We now give a second construction of inequalities that,
if satisfied, ensure that \GMk model parameters on a 3-leaf tree are stochastic.
While in \S \ref{sec:THREEleaf} we constructed matrices encoding positivity
of parameters through requirements on associated quadratic forms, 
here we instead construct matrices whose eigenvalues encode 
parameters. 

\medskip

Suppose that $P = \psi_T ( \boldsymbol \pi, \{M_1, M_2, M_3\})$ for
complex nonsingular $\{ \boldsymbol \pi, M_1 M_2, M_3\}$. Recall that
for $i \in[k]$, 
\begin{align*}
P_{i\cdot\cdot}&=P*_1\mathbf e_i= M_2^T  \, \diag(\boldsymbol \pi) \Lambda_{1,i} \, M_3, \\
P_{+\cdot\cdot}&=P*_1\mathbf 1= M_2^T \, \diag(\boldsymbol \pi) M_3
\end{align*}
where $\Lambda_{1,i}=\diag(M_1 \mathbf e_i)$ is the diagonal matrix with entries from 
the $i$th column of $M_1$. Thus, 
by nonsingularity of the parameters,
$$A_{1,i} := P_{+\cdot\cdot}^{-1}P_{i\cdot\cdot}  = M_3^{-1}\Lambda_{1,i}M_3$$
has as eigenvalues the entries of the $i$th column
of $M_1$.  (This construction underlies the proof of parameter
identifiability for the GM($k)$ model on trees \cite{Chang96}, and the construction 
of phylogenetic invariants in \cite{AR03}, including the equality in condition (i)
of Theorem \ref{thm:AJRS} of this paper.)  Similarly we define
the matrices 
\begin{align*}
A_{2,i} :=   P_{\cdot+\cdot}^{-1}P_{\cdot i\cdot} &= M_3^{-1} \, \Lambda_{2,i} \, M_3,\\
A_{3,i} :=  P_{\cdot\cdot +}^{-1}P_{\cdot\cdot i} &= M_2^{-1} \, \Lambda_{3,i} \, M_2
\end{align*}
with $\Lambda_{j,i}=\diag(M_j\mathbf e_i)$ the diagonal matrix with entries from the
$i$th column of the matrix $M_j$.

\begin{proposition}\label{prop:SturmSeqToMarkov}
  Let $P=\psi_T ( \boldsymbol \pi, \{M_1, M_2, M_3\})$ be a real
  $k\times k\times k$ tensor that is the image of complex nonsingular
  parameters. For each of the matrices $A_{j,i}$, $j\in\{1,2,3\}$,
  $i\in [k]$, assume the characteristic polynomial has neither $0$ nor $1$
  as a root, and let $\mathcal S_{j,i}$ denote the standard Sturm
  sequence for it.  Then $V_{\mathcal S_{j,i}}(0) = k$ and
  $V_{\mathcal S_{j,i}}(1) = 0$ for all $i,j$ if, and only if, the
  matrices $M_1$, $M_2$, $M_3$ are positive Markov matrices with no
  repeated element in any column and $\boldsymbol \pi$ is real.
\end{proposition}

\begin{proof}
The statement about the $M_i$ follows from Corollary \ref{cor:signs}. 
Then, since the $M_i$ are real and invertible
$$\boldsymbol \pi = ((P \cdot (M_1^{-1}, M_2^{-1}, M_3^{-1})) \ast_1 
\mathbf 1) \ast_2 \mathbf 1$$
shows
$\boldsymbol \pi$ is real. 
\qquad\end{proof}

\smallskip

Each matrix $A_{j,i}$ has entries that are rational in the entries of $P$, with 
denominator $\det(P*_j \mathbf 1)$ of degree $k$ and numerator of degree $k$.
The characteristic polynomial $f$ of $A_{j,i}$ thus also has coefficients that
are rational functions in $P$, and in fact the non-leading coefficients $c_i$ of
$f$ are rational of degree $k$ over $k$. This can be seen
explicitly, for example when $j=1$, in the following:
\begin{align}\label{degCoeffs}
f(x) &= \det( xI - \Psbc^{-1} P_{i\cdot\cdot})\notag \\
&= \det( x \Psbc^{-1}{\Psbc} - \Psbc^{-1} P_{i\cdot\cdot})\notag \\ 
&= \det\big( \Psbc^{-1} (x {\Psbc} - P_{i\cdot\cdot})\big)\notag \\
&= \frac{\det\big( x {\Psbc} -  P_{i\cdot\cdot} \big)}{\det(\Psbc)}. 
\end{align}

It follows that the Sturm sequence inequalities, which
  are constructed from the coefficients $c_i$, are rational in the
  entries of $P$ as well.  Indeed, by multiplying each of these
  inequalities by a sufficiently high even power of $\det(P*_j \mathbf
  1)$ to avoid changing signs, these expressions become polynomial in
  $P$.  Thus, one can phrase the conditions $V_{\mathcal S_{j,i}}(0) =
  k$ and $V_{\mathcal S_{j,i}}(1) = 0$ as a collection of polynomial
  inequalities. Finally, note that since $f$ is a monic characteristic
  polynomial, then $(-1)^kf_0(0)=\det(A_{j,i})>0$ and $f_k(0) =
  f_k(1)$, so the signs of all $f_j(0)$ and $f_j(1)$ are
  determined. This leads to the following: 

\begin{corollary}\label{cor:SturmIneq}
  Consider real $k \times k\times k$ tensors $P=\psi_T ( \boldsymbol
  \pi, \{M_1, M_2, M_3\})$ arising from complex nonsingular parameters.  Then
  one can give a finite collection of strict polynomial inequalities
  that hold precisely when the $M_i$ are positive Markov matrices with
  no repeated column entries, and $\boldsymbol\pi$ is real.
\end{corollary}

\smallskip

\emph{Remark.}  For simplicity, we have restricted our discussion to polynomials 
with no multiple roots, leading to the constraints on the matrix
column entries given above.  However, this restriction can 
be removed, by considering Sturm sequences for polynomials with
repeated roots.  
We suggest \cite{KKsturm, JMsturm} for futher information. 

\smallskip

Note that the non-strict versions of the inequalities of Corollary \ref{cor:SturmIneq} must
continue to hold on the closure of the image of parameters described
in the corollary.  As this closure includes the image of 
Markov matrices which may have repeated column entries, or entries of 0 or 1, or are singular, the
non-strict inequalities hold for \emph{all} stochastic parameters. 

However, some distributions arising from non-stochastic
parameters may satisfy the non-strict inequalities as well. 
Thus while we have obtained semialgebraic statements guaranteeing
stochasticity of nonsingular parameters of a particular form, this does not
seem to lead to a complete semialgebraic description of
the image of all stochastic parameters for arbitrary $k$.
In the case $k=2$, however, we can do better, as we show next.

\subsection{Sturm sequences for GM($2$)}

In the specific case of the $2$-state model, we explicitly give the inequalities of Corollary \ref{cor:SturmIneq}. 
To this end, suppose that $P=\psi_T(\boldsymbol \pi, \{M_1,M_2,M_3\})$ 
is a $2 \times 2 \times 2$ probability distribution arising
from nonsingular parameters.

For any $j \in\{1,2,3\}$, $i\in [2]$, let 
$A=A_{j,i}$.
The standard Sturm sequence of the characteristic polynomial of $A$
is 
\begin{align*}
f_0   &={} t^2 -\tr(A) t + \det(A),\\
f_1  &={} 2t -\tr(A),\\
f_2  &={}  \frac{1}{4}({\tr (A)}^2-4\det(A)) = \frac14{ \delta(f_0)},
\end{align*}
where $\delta (f)$ denotes the discriminant of the monic quadratic polynomial
$f$. 

Since $k=2$, the nonsingularity of the $M_i$ together
with the fact that their rows sum to 1 implies that none of their
columns have repeated entries. Thus the characteristic polynomial
$f_0$ must have distinct roots.  To ensure that these roots
are in $(0,1)$, so that the matrix parameters are Markov, 
the variation in signs of the Sturm sequence must be: 
\begin{align}
f_0(0) &= \,\det(A) > 0, &f_0(1)=1 -\tr(A)  + \det(A) > 0, \notag\\
f_1(0) &= -\tr(A) < 0,  &f_1(1)=2 -\tr(A) > 0, \label{eq:sturm2}\\
f_2(0) &= ~\,\frac{1}{4} \delta(f) > 0,   &f_2(1)=\frac{1}{4} \delta(f) > 0. \notag
\end{align}

Each of these inequalities can be expressed using rational expressions in the entries of $P$.
For instance, for $j=1$, $i=1$, the first two inequalities ensuring that the first column of $M_1$ 
has distinct entries between 0 and 1 are 
\begin{equation}
f_0(0) = \det (A) = \frac{\det (P_{1 \cdot \cdot})}{\det(\Psbc)} > 0,\label{eq:f0ineq}
\end{equation}
and
$$f_1(0) = -\tr(A) =  \frac{-2 \det(P_{1 \cdot \cdot}) + p_{112} p_{221}-p_{211} p_{122}-p_{111} p_{222}+p_{212} p_{121}}{\det(\Psbc)} < 0.
$$
After multiplying each of the inequalities of \eqref{eq:sturm2} by $\det(\Psbc)^2$ 
to clear denominators, we obtain five distinct  polynomial inequalities of degree $4$
in the entries of $P$, as well as one degree-2 inequality $$\det(\Psbc)\ne0.$$  Note too that $f_2(0) = f_2(1)$ is a 
positive multiple of the discriminant of $f_0$, and its positivity guarantees (again) 
that the roots of $f_0$ are distinct and real. 

The inequality \eqref{eq:f0ineq} has a direct statistical
interpretation: Assuming the states of the variables $X_i$ are encoded
with numerical values $s$ and $s+1$, then $\det(\Psbc) = \Cov(X_2,
X_3)$ and $\det (P_{1 \cdot \cdot})$ is a positive multiple of $\Cov
(X_2, X_3 \mid X_1 = 1)$. Thus the inequality states that the sign of
the association of $X_2$ and $X_3$ is the same whether we have
information about $X_1$ or not.  Viewing the $3$-leaf tree as a
graphical model for nonsingular parameters, this should be expected,
but that it arises cleanly from Sturm theory is a pleasant surprise.

\smallskip

Of additional interest is the observation that $f_2$ can be expressed in terms of the
hyperdeterminant $\hyper$: 
$$
f_2 = \frac{ \delta (f)}4  = \frac{\hyper }{4\det(\Psbc)^{2}} .
$$
Since $\Delta(P)\ne 0$ implies $\det(\Psbc)\ne0$, we find
\begin{equation}\label{eq:hyperDelta}
f_2>0 \text{\; if, and only if, \;} \Delta(P)>0.
\end{equation}
In particular, using 
Theorem \ref{thm:dSL}, we see the Sturm inequality involving $f_2$ implies 
that $P$ arises from nonsingular real parameters, and thus an additional assumption of that fact is 
not needed to supplement the Sturm inequalities.

\smallskip

We further note that the inequalities \eqref{eq:sturm2} for various $A_{j,i}$ are not independent of one another.
Since $A_{j,1}+A_{j,2}=I$, it follows that the two matrices $A_{j,1}$, $A_{j,2}$ give rise to the same inequalities.

\smallskip

The inequalities \eqref{eq:sturm2}, unfortunately,  are not
sufficient to ensure the root distribution $\boldsymbol \pi$ is also
positive. For instance,
\begin{align*}
P &= 
\left[
\begin{array}{cc|cc}
   0.65439 &  0.07191 &   0.16361 & 0.01809\\
   0.07191 &  0.00079 &   0.01809 & 0.00121
\end{array}
\right]
\end{align*}
is a probability distribution that satisfies $\hyper > 0$ and the
Sturm inequalities for each $A_{j,i}$, but the root
distribution $\boldsymbol \pi = (1.01, -0.01)$ is not stochastic.

Nonetheless, in the $2$-state case we can construct another inequality in $P$ that ensures the 
root distribution is positive.  If $P \in \image (\psi_T)$ for nonsingular real parameters, then
\begin{equation*}
\frac{\det(\Psbc) \det(\Pasc) \det(\Pabs)}{\hyper} = \pi_1 \pi_2.
\end{equation*}
This is easily verified using transformation properties of the determinants and $\Delta$ under
the action of $(M_1,M_2,M_3)$ on $\Diag(\boldsymbol \pi)$.
(See \cite[p.~136]{ARsbd} for an earlier derivation and application of this equation.)
Since $\pi_1+\pi_2=1$, the positivity of the $\pi_i$ is equivalent to
\begin{equation*}\label{eq:pistochastic}
0 < \frac{\det(\Psbc) \det(\Pasc) \det(\Pabs)}{\hyper} < \frac{1}{4}.
\end{equation*}
Moreover, because $\hyper>0$, this in turn is equivalent to the inequalities
\begin{equation}0<\det(\Psbc) \det(\Pasc) \det(\Pabs)<\frac 14\hyper.\label{eq:pipos}\end{equation}
Although the second inequality here is not homogeneous, it can be made 
homogeneous of degree 6 by multiplication of the right side by 
$1=(\sum_{i,j,k=1}^2 P_{ijk})^2$.

Putting this all together, we have an alternative semialgebraic test, 
to be contrasted with case 1 of Theorem \ref{thm:pos}, 
for testing  that $\boldsymbol \pi$ is stochastic.

\begin{proposition}
  The image of the positive nonsingular parameterization map for
  GM(2) on a 3-leaf tree can be characterized as the probability
  distributions satisfying an explicit collection of strict
  polynomial inequalities: $3$ of degree $2$, $13$ of degree $4$, and $2$ of degree $6$.
  
  \end{proposition}

\begin{proof}
By the above discussion, the degree-$2$  inequalities  $\det(P*_i \mathbf 1)\ne 0$ and the  5 degree-$4$ inequalities arising from \eqref{eq:sturm2} for each of the 3 choices of ${j}$ suffice to ensure 
that nonsingular parameters exist, with Markov matrices having positive entries. Only 13 of these degree-$4$ inequalities are distinct, as each $j$ leads to $\hyper>0$.
Then the degree-$6$ inequalities of \eqref{eq:pipos} ensure $\boldsymbol \pi$ has positive entries.
 \qquad\end{proof}

Note that case 1 of Theorem \ref{thm:pos} gave a description using 3 degree-2, 1 degree-$4$, 4 degree-$5$, and 4 degree-$6$ polynomials.
The description arising from Sturm theory thus uses fewer degree-5 and 6 polynomials, but more degree-4 ones.

\subsection{Sturm sequences for GM($3$)}    
\medskip

We now give several examples of Sturm sequence inequalities for 
GM($3$) on the $3$-taxon tree.

If $A$ is a $3\times 3$ matrix with positive determinant and characteristic polynomial 
$f(x) = x^3 + c_2 x^2 + c_1 x + c_0$
without roots at 0 or 1, then $A$ has 3 distinct eigenvalues in the interval (0,1) if, and only if,
\begin{align}
f_0(0) &= \,c_0  &< 0,                                  &&f_0(1)&= \, 1 +c_2  + c_1 +c_0 &> 0, \notag\\
f_1(0) &= \, c_1 &> 0,                                      &&f_1(1)&= \, 3+2c_2+c_1 &> 0,\notag\\
f_2(0) &= -c_0 + \frac{1}{9} c_1 c_2 &< 0,          &&f_2(1)&= -\frac{2}{3}c_1 + \frac{2}{9} c_2^2-c_0+\frac{1}{9}c_1c_2 &> 0, \label{eq:sturm3}\\
f_3(0) &= \, \frac{9}{4} \frac{\delta(f)}{{(3c_1-c_2^2)}^2}&> 0, &&f_3(1) &= \, \frac{9}{4} \frac{\delta(f)}{{(3c_1-c_2^2)}^2}&> 0,\notag
\end{align}
where $\delta(f)$ denotes the discriminant of $f$. 
Here, of course, $c_0=-\det(A)$, $c_2=-\tr(A)$ and $c_1$ is quadratic in the entries of $A$.
However, by \eqref{degCoeffs}, if $A=A_{j,i}$ then each $c_i$ is rational in the entries of $P$, with numerator of degree $3$ and denominator $\det (P*_j \mathbf 1)$.

By multiplying the top 6 inequalities of \eqref{eq:sturm3} by $\det (P*_j \mathbf 1)^2$, 
we obtain six polynomial inequalities of degree $6$ in $P$.
In the case that $A = A_{1,1} = \Psbc^{-1} P_{1\cdot\cdot}$, for example,
the first inequality is
$$0<-\det(\Psbc)^2 f_0(0) = \det(\Psbc) \det(P_{1\cdot\cdot}).$$
By \eqref{degCoeffs} we know that $\det(\Psbc)^2 c_1$ and $\det(\Psbc)^2 c_2$ are polynomials
of the form $\det(\Psbc) K$, 
where $K$ is homogeneous of degree $3$ in the entries of $P$. Computations with
Maple show $K$
has $42$ monomial summands for $c_1$, and $114$ monomial summands for $c_2$.  

The inequality of the bottom row of \eqref{eq:sturm3} can be simplified to $\delta(f)>0$, which is of degree $4$ in the $c_i$.  Thus, $\det (\Psbc)^4 \delta (f)>0$
is a polynomial inequality of degree $12$ in $P$.
We omit writing these Sturm inequalities explicitly.

Instead, we illustrate a more direct application of the relevent Sturm theory, in which the semialgebraic description of the model is present only implicitly.
Consider the probability distribution with exact rational entries given by

\begin{align}\label{ex:complexParams}
P &= 
\tiny{\left[
\begin{array}{ccc|ccc|ccc}
    0.1500  &  0.0130  &  0.105\bar 3  &  0.0130   & 0.0050  &  0.015\bar 3  &  0.105\bar 3 &   0.015\bar 3  &  0.077\bar 6\\
    0.0130  &  0.0050  &  0.015\bar 3  &  0.0050   & 0.0090  &  0.009 \bar 3  &  0.015\bar 3 &   0.009 \bar 3  &  0.018\bar 6\\
    0.105\bar 3  &  0.015\bar 3  &  0.077 \bar 6  &  0.015\bar 3   & 0.009 \bar 3  &  0.018\bar 6  &  0.077 \bar 6 &   0.018\bar 6  &  0.0620
\end{array}
\right]}.
\end{align}

One can check that $P$ satisfies the conditions of Theorem \ref{thm:AJRS}, so
$P$ is in the image of nonsingular complex parameters.  Then the values
of the Sturm sequence at $x=0$ and $x=1$ for the characteristic polynomial of $A_{1,1}$ are approximately as in Table \ref{tb:sturm}.
Thus the sign variations are $V_{\mathcal S}(0) = 2$ and 
$V_{\mathcal S}(1) = 1$, so the first column of
$M_1$ has exactly $1$ distinct real entry. This then implies that either $M_1$ is not real, or that its first column contains the same entry in all rows. Since the second possibility implies that a $3\times 3$ Markov matrix is singular, we can conclude that $P$ does not arise from stochastic
parameters.  

\begin{table}[h]
\begin{center}
\begin{tabular}{c|cccc}
& \hskip 2mm $f_0(x)$ & $f_1(x)$ & \hskip 2mm $f_2(x)$ & $f_3(x)$\\
\hline\hline
$x=0$ &  $-0.1087$ & $0.7225$ & $-0.0117$ & $-0.0283$\\
$x=1$ &  \hskip 3mm $0.1138$ & $0.7225$ & \hskip 3mm $0.0067$ & $-0.0283$
\end{tabular}
\end{center}
\caption{Sturm sequence values for the characteristics polynomial of $A_{1,1}$ for the tensor $P$ of  \eqref{ex:complexParams}}\label{tb:sturm}
\end{table}

This example did not actually require the full strength of
Sturm's theorem; it is sufficient to note that the discriminant of the cubic
characteristic polynomial is negative to conclude that the first
column of $M_1$ has two complex entries, and one real one.  This is 
special to the small size of the state space, however. For the case
of most interest in phylogenetics, $k=4$, the sign of the quartic discriminant alone
does not carry enough information to determine whether all roots of a polynomial are real. Moreover, even for $k=3$, the Sturm sequence is needed to ensure roots are between 0 and 1. 

One might optimistically hope that some $k=3$ analog of the hyperdeterminant, such as those in \cite{AJRS2012}, might arise easily from Sturm theory, as the hyperdeterminant itself did in the case $k=2$. Unfortunately this does not appear to be the case, at least by the most straightforward considerations.
\smallskip

In closing, we note that for $k\ge3$ we have not given in this section
any condition to ensure that $\boldsymbol \pi$ has positive entries.
For $k=2$ we did do so, using the transformation formula for $\hyper$, but
this idea does not seem to
generalize.  The only way we know to obtain a semialgebraic condition ensuring this is through
the quadratic form approach used in \S \ref{sec:THREEleaf}.

\section{GM($k$) on $n$-leaf trees}\label{sec:nleaf}

We now extend the results of the previous sections to $n$-leaf 
trees, for $n>3$.
To vary the choice of the root node of the tree in our arguments, we 
need the following. Similar lemmas 
are given in
\cite[Theorem 2]{SSH94} and \cite[Proposition 1]{AR03}. 

 \begin{lemma} \label{lem:moveroot}
 Suppose stochastic parameters are
   given for the GM($k$) model on a tree $T$ with the root located at a
   specific node of $T$.  Then there are stochastic parameters for $T$
   rooted at any other node of $T$, or at a
   node of valence $2$ introduced along an edge of $T$,
   which lead to the same distribution. Moreover, if the original 
   parameters were nonsingular and/or positive, then so are the new ones.  
   \end{lemma}

\begin{proof} It is enough to show this on a tree with a single edge,
  as one may then successively apply that result along the edges in a path in a larger tree.  
  
  We show first that the root
  may be moved from one vertex of an edge to the other. For this
  it is sufficient to show that
 given any probability distribution $\boldsymbol \pi$ and Markov matrix
  $M$, there exists a probability distribution
  $\tilde{\boldsymbol \pi}$ and Markov matrix $\widetilde M$ with
  $\diag(\boldsymbol \pi )M=P=\widetilde M^T \diag(\tilde{\boldsymbol
    \pi} )$. This is straightforward if the column sums of $P$ are
  nonzero. If a column sum of $P$ is zero, and hence all entries in
  the column are zero, then the corresponding entry of
  $\tilde{\boldsymbol \pi}$ must be zero while that row of $\widetilde
  M$ can be arbitrary.  If the original parameters were nonsingular or positive, 
  then showing that the new ones are as well is straightforward.
  
  If instead we wish to move the root from vertex $a$ on edge 
  $(a,b)$ to a new internal node $r$ introduced to subdivide the edge,
  first introduce $r$ and let $M_{(a,r)}=M_{(a,b)}$ and $M_{(r,b)}=I$. Then 
  move the root from $a$ to $r$ as above. For the case of positive parameters, 
  instead pick $M_{(r,b)}$ to have positive entries but be near enough to 
  $I$ that $M_{(a,r)}=M_{(a,b)}M_{(r,b)}^{-1}$ has positive entries.
  \qquad\end{proof}

\smallskip
 
Note that for real or complex parameters Lemma \ref{lem:moveroot}
fails to hold as the examples $\boldsymbol \pi=(1/2, 1/2)$,
$M=\begin{pmatrix} s&1-s\\2-s&s-1\end{pmatrix}$, $s\ne 1$ show.  (The
problem is simply that a column sum of $\diag(\boldsymbol \pi)M$ can
be zero though the column is not the zero vector.)  However, if the
parameters are nonsingular, we can still move the root by modifying
the above argument.  Indeed, nonsingularity of parameters implies that
from $\diag(\boldsymbol \pi )M=\widetilde M^T \diag(\tilde{\boldsymbol
  \pi} )$ one can solve for a nonsingular $\widetilde M$, since the
other three matrices in the equation are nonsingular.  This shows the
following.

 \begin{lemma} \label{lem:moveroot2} 
 Suppose real or complex nonsingular parameters are
   given for the GM($k$) model on a tree $T$ with the root located
   at a
   specific node of $T$.  Then there are nonsingular parameters for $T$
   rooted at any other node of $T$, or at a
   node of valence $2$ introduced along an edge of $T$,
   which lead to the same distribution.
\end{lemma}

\smallskip
  
We now show that independent subsets of variables allow the question
of determining if a distribution arises from parameters on a tree to
be `decomposed' into the same question for the marginalizations to the
subsets.

\begin{proposition} \label{prop:patchtogetherindependentsets} 
Let $P$
  be a joint distribution of a set $L$ of $k$-state variables such
  that for some partition $L_1 | L_2 | \cdots | L_s$ of $L$, the
  variable sets $L_i$ and $L_j$ are independent for all $i\ne
  j$. Suppose the marginal distribution of each $L_i$ arises from
  nonsingular GM($k$) parameters on a tree $T_i$. Then $P$ arises from
  GM($k$) parameters on any tree $T$ which can be obtained by connecting
  the trees $T_1,T_2,\dots, T_s$ by the introduction of new edges
  between them (with endpoints possibly subdividing either edges of
  the $T_i$ or previously introduced edges joining some of the $T_i$).
 \end{proposition}

 Note that the converse of this statement, that if $P$ arises from
 parameters for the GM($k$) model on an $|L|$-leaf tree then the
 marginal distributions of each $L_i$ arise from parameters for the
 GM($k$) model on an $|L_i|$-leaf tree, is well-known, and does not
 require the independence of the variable sets, or nonsingularity of
 parameters.
  
\begin{proof}
  It is enough to consider a partition of $L$ into two independent
  subsets, $L_1|L_2$.  Let $T$ be any tree formed by connecting $T_1$
  and $T_2$ by a single edge, possibly with endpoints introduced to
  subdivide edges of one or both of the $T_i$.  If $e = (r_1,r_2)$ is
  the edge joining $T_1$ and $T_2$, with $r_i$ in $T_i$, then by Lemma
  \ref{lem:moveroot2} we may assume that parameters on $T_1$ and $T_2$
  are given for roots $r_1$ and $r_2$.  We root $T$ at $r_1$ and then
  specify parameters on $T$ as the root distribution $\boldsymbol
  \pi_1 $ for $T_1$, all matrix parameters on the edges of $T_1$ and
  $T_2$, and for the edge $e$ the matrix $M_e = \mathbf 1
  {{\boldsymbol \pi}_2}^T$ where $\boldsymbol \pi_2$ is the root
  distribution on $T_2$.
   
   Let $\tilde P$ denote the image of these parameters under  $\psi_T$.
   The edge $e$ of $T$ induces the split $L_1 | L_2$ of the leaf variables, and flattening
   with respect to $e$ gives $\Flat_e(\tilde P)= A^TCB$ where $A,B$ are 
   $k \times k^{\vert L_1 \vert}$ and $k \times k^{\vert L_2 \vert}$ matrices
   depending only on the matrix parameters on the subtrees $T_1$ and $T_2$, 
  and 
  \begin{align*}
  C &= \diag(\boldsymbol \pi_{1}) M_{e}, \\
     &= \diag( \boldsymbol \pi_{1})\mathbf 1 \boldsymbol \pi_{2}^T = \boldsymbol \pi_{1} \boldsymbol \pi_{2}^T.
  \end{align*}
  Indeed, in the stochastic case, 
  $A$ gives probabilities of observations at the leaves $L_1$ conditioned on the state at $r_1$, 
  $B$ gives probabilities of observations at the leaves $L_2$ conditioned on the state at $r_2$, 
  and $C$ is a matrix giving the joint distribution of states
  at $r_1$ and $r_2$.
  Observing that $A^TCB=(A^T \boldsymbol \pi_{1})(\boldsymbol \pi_{2}^TB)$, independence implies that $\tilde P$
is the product of the same marginal distributions on $L_1$ and $L_2$ as $P$, and
hence $\tilde P=P$.
 \qquad\end{proof}

\smallskip

One can replace the assumption of nonsingularity of parameters in
this proposition with one of stochasticity, since the main technical
point in the proof is that we need freedom to move the root to any
node of valence $2$ along an edge of $T$.  By Lemma
\ref{lem:moveroot}, this holds for stochastic parameters, so we obtain
the following.

\begin{proposition} \label{prop:patchtogetherindependentsets2}
  With $P$ and the $L_i$ as in Proposition \ref{prop:patchtogetherindependentsets}, if
  the marginal distributions of each $L_i$ arise from stochastic parameters 
  for the GM($k$) model on an
  $|L_i|$-leaf tree, then $P$ arises from stochastic parameters for the GM($k$) model
  on an $|L|$-leaf tree. If the parameters on the $|L_i|$-leaf trees are positive, so are those on
  the $|L|$-leaf tree.
\end{proposition}

\smallskip

By this proposition, the only sets we must understand to build a
semialgebraic description for the full $n$-leaf stochastic  model are
the image of parameters for $m$-leaf trees, $m\le n$, when no
subsets of the $m$ leaf variables are independent.  In the case $k=2$, by
Proposition \ref{prop:ind}, this is precisely the images of
nonsingular parameters.  Unfortunately, for larger $k$ characterizing
such images is much more difficult,  and we do not accomplish it in
this paper. 

One way to think about the difficulty with $k>2$ is in terms of matrix rank. 
When $k=2$, a Markov matrix has either rank 1 or rank 2, which results in 
independent variables or nonsingular parameters, respectively. For larger $k$, 
a Markov matrix may have rank between the extremes of 1 and $k$, which 
again produce independent variables or nonsingular parameters.  These intermediate 
cases of singular Markov matrices that are not of rank 1 would each require a 
detailed analysis, both in the 3-leaf tree case, and then to produce 
some type of decomposition for larger trees.

Rather than pursue a detailed semialgebraic model description allowing all ranks 
of Markov matrices for all $k$, we
instead choose to focus on the image of nonsingular parameters. These are certainly 
the most important in most applications. Moreover, any distribution arising from 
singular parameters will lie in the topological closure of this set, as singular parameters 
can be approximated by nonsingular ones. Of course the closure may also contain points 
that do not arise from any parameters, so this does not circumvent the difficulties of 
dealing with the many special cases of intermediate rank if an exact semialgebraic 
description of the full stochastic model is desired.

\smallskip
In the setting of nonsingular, though not necessarily stochastic parameters, we obtain the following.
 \begin{proposition}\label{prop:nonsing3toN}
   Let $P$ be an $n$-dimensional $k\times k \times \cdots \times k$ 
   distribution  with $n\ge 3$.
   Then $P$ arises from nonsingular complex parameters on a binary tree $T$ if,
   and only if,
 \begin{romannum}
 \item All marginalizations of $P$ to $3$ variables arise from 
   nonsingular parameters on the induced $3$-leaf, $3$-edge trees,
   and \label{cond:1}
 \item For all internal edges $e$ of $T$, all $(k+1) \times (k+1)$ minors of the matrix flattening $Flat_e(P)$ are
   0.\label{cond:2} 
    \end{romannum}
    
  Moreover, such nonsingular parameters are unique up to label
   swapping at internal nodes of $T$.
 \end{proposition}
 
 \smallskip
 
 Note that condition (i) can be stated in terms of explicit semialgebraic conditions,
 using Corollary \ref{cor:complexToMarkov}.  
 Also, the polynomial equalities of condition (ii) are usually called \emph{edge invariants} \cite{ARgm}.
 
 \smallskip
 
 \begin{proof}
   For the forward implication, condition (i) follows since
   marginalizations arise from the model on the
   associated induced subtree, using Markov matrices that are products 
   of the original ones. Item (ii) is from \cite{ARgm}, where it
   is shown that all $P \in \image ( \psi_T)$ satisfy the edge invariants.  
   (The nonsingularity of parameters is not required for either of these.)  
   
   For the reverse implication, we proceed by induction on the size $n$ of the 
   variable set $L$. The claim holds by assumption in the base case of  $n=3$.
   Assume the statement is true for fewer than $n\ge 4$ variables. We identify leaves 
   of $T$ with the variables associated to them. Choose some internal edge 
   $e_0=(a,b)$ of $T$, corresponding
   to the split $L_1|L_2$ of $L$, with $|L_1|,|L_2|\ge 2$, $a$ in the subtree
   spanned by $L_1$, and $b$ in the subtree spanned by $L_2$. Introducing a vertex $c$
   subdividing $(a,b)$, let $T_1$ be the subtree with leaves $L_1\cup\{c\}$ and
   $T_2$ the subtree with leaves $L_2\cup\{c\}$. Thus $(a,c)$ in $T_1$ and
   $(b,c)$ in $T_2$ are the edges formed from dividing $(a,b)$.

   Since the edge invariants are satisfied by $P$, $\Flat_{e_0}(P)$ has
   rank at most $k$. Therefore, there exist $k^{|L_1|}\times k$ and
   $k \times k^{|L_2|}$ matrices $A,B$, both of rank at most $k$, with
$$\Flat_{e_0}(P)=AB.$$
Choose a single variable $\ell_2 \in L_2$ and let $Q$ denote the
marginalization of $P$ to $L_1\cup\{\ell_2\}$. Then there is a
$k^{|L_2|}\times k$ matrix $N$ such that
$$\Flat_{e_0}(P)N=ABN=\Flat_{e_1}(Q),$$
where this last flattening is along the edge $e_1 = (a,\ell_2)$ in the induced
subtree on $L_1\cup\{\ell_2\}$. Stated differently, multiplication by $N$ marginalizes over all
those leaves in $L_2$ except $\ell_2$. 

Since $Q$ also satisfies
conditions (i) and (ii), by the inductive hypothesis $Q$ arises from
nonsingular parameters. 
Moreover, we see that $\Flat_{e_1}(Q)$ has rank $k$,
since marginalization over all but one variable in $L_1$
is seen to produce a rank $k$ matrix from the nonsingular parameterization.
It follows that the $k \times k$
matrix $BN$ has rank $k$.  Replacing $A$ and $B$ with $AC$ and
$C^{-1}B$ for some invertible $k \times k$ matrix $C$, we may further
assume the rows of $BN$ add to $1$.

Now since $Q$ arises from nonsingular parameters on a $(|L_1|+1)$-leaf
tree isomorphic to $T_1$ rooted at $a$, we claim that
$Q'=Q*_{\ell_2}(BN)^{-1} $ arises from nonsingular complex parameters
on $T_1$ for some suitable choice of $B$. Indeed, $Q'$ arises from the
same parameters as $Q$, except that on the edge $(a,c)$ we use the
matrix parameter that is the product of the one on the edge leading to
$\ell_2$ and $(BN)^{-1}$.  Since $(BN)^{-1}$ is a nonsingular matrix
with rows summing to one, the only condition to check is that the
marginalization of the resulting distribution to $c$ has no zero
entries.  But this marginalization is $\mathbf v_c=\mathbf
v_{{\ell_2}} (BN)^{-1}$, and has a zero entry only if $\mathbf
v_{{\ell_2}}$ is in the left nullspace of one (or more) of the columns
of ${(BN)}^{-1}$.  However, replacing $A$ and $B$ with $AC$ and
$C^{-1}B$ for some appropriate nonsingular matrix $C$ whose rows sum
to one, we can ensure that $\mathbf v_c$ has no zero entries.

Since the parameters producing $Q'$ are nonsingular, 
by Lemma \ref{lem:moveroot2} we may reroot $T_1$ at $c$, with parameters 
the root distribution $\mathbf v_c$, matrices
$\{M_{e}\}$ on all edges of $T_1$ corresponding to ones in $T$, and
matrix $M_{(c,a)}$ on the edge $(c,a)$.

Now with $K$ the matrix which marginalizes $\Flat_{e_0}(P)$ over all
elements of $L_1$ but one, say $\ell_1$, we see $$K \Flat_{e_0}(P)
=KAB=\Flat_{e_2}(U),$$ where $U$ is the marginalization of $P$ over the
same elements of $L_1$ and the last flattening is on $e_2 = (b, \ell_1)$ in the
induced subtree, which is isomorphic to $T_2$.  But by induction $U$
arises from nonsingular parameters on $T_2$ rooted at $b$. Letting $M$ be
the product of the matrix parameters on the edges in the path from $c$
to $\ell_1$ in $T_1$.  
Then
$U'=U*_{\ell_1} M^{-1}$ also arises from nonsingular parameters on
$T_2$ (checking that its marginalization to $c$ is $\mathbf v_{\ell_1} M^{-1}
=\mathbf v_c$, which has no zeros by construction).

Note now that $U'$ has flattening $(M^{-1})^TKAB$.  But $(M^{-1})^TKA=\diag(\mathbf v_c)$ by
construction.  Thus $\diag(\mathbf
v_c)B$ is the $c|L_2$ flattening of a tensor arising on $T_2$ from
nonsingular complex parameters. With the root at $c$,  let $M_{e}$ be the
Markov parameters for all edges of $T_2$ corresponding to ones in $T$,
and $M_{(c,b)}$ the Markov matrix on $(c,b)$. The root distribution 
$\mathbf v_c$ is the same as for $T_1$.

It remains to check that $P$ is the image of the parameters on $T$
with subdivided edges $(c,a)$ and $(c,b)$ rooted at $c$ given by
$\mathbf v_c$, $\{M_{e}\}_{e\ne (a,b)}$, and $M_{(c,a)}$ and
$M_{(c,b)}$.  But these parameters lead to the distributions $Q'$ and
$U'$ on $T_1$ and $T_2$ respectively. Since $\Flat_{(a,c)}(Q')=A$
while $\Flat_{(c,b)}(U')=\diag(\mathbf v_c)B$, the equation
$\Flat_{e}(P)=AB$ shows they produce $P$ on $T$.

That the parameters are unique, up to label swapping at
the internal nodes of $T$, follows from the 3-leaf case.
 \qquad\end{proof}

\smallskip

Note that in establishing the reverse implication in Proposition
\ref{prop:nonsing3toN} we did not use condition (i) for every
3-variable marginalization. Informally, given a tree $T$ one could
choose a sequence of edges which can be successively `cut' (by the
introduction of the node $c$ in the inductive proof above) to produce
a forest of 3-taxon trees. Then condition (i) is only needed for a
subset of the 3-leaf marginalizations, determined by the sequence of
edges chosen to cut and the choice of the variables denoted
$\ell_1,\ell_2$ in the proof. Similarly, not all edge flattenings of
condition (ii) are used: For the first edge to be `cut', one uses the
full edge flattening, but after that, only edge flattenings of
marginalizations to subsets of variables are needed. Thus the full set
of conditions given in this proposition is actually equivalent to a
subset of them, though pinning down a precise subset is rather messy
and will not be pursued here.

\medskip

Supposing now that an $n$-dimensional distribution $P$ arises from
nonsingular complex parameters on a binary tree $T$, we wish to give
semialgebraic conditions that are satisfied if, and only if, the
parameters are stochastic.  By considering only marginalizations to
$3$ variables and appealing to Proposition \ref{thm:kstate-quadform},
we can give conditions that hold precisely when the root distribution
and products of matrix parameters along any path leading from an
interior vertex of $T$ to leaves are stochastic. This immediately
yields semialgebraic conditions that the root distribution and matrix
parameters on terminal edges are stochastic. However, additional
criteria are needed to ensure matrix parameters on interior edges are
stochastic. In the 4-leaf case, such criteria are given by the
following.

\begin{proposition} \label{prop:4taxa} Suppose a tensor $P$ arises
  from nonsingular complex parameters for GM($k)$ on the $4$-leaf tree
  $12|34$.  If the $3$-marginalizations $P_{\cdot\cdot\cdot+}$ and
  $P_{+\cdot\cdot\cdot}$ arise from stochastic parameters and, in
  addition, all principal minors of the $k^2\times k^2$ matrix
\begin{equation}\det(P_{+\cdot\cdot+}  )\det(P_{\cdot+\cdot+} )\Flat_{13|24}\left( P*_2 (\adj(P^T_{+\cdot\cdot +}) P^T_{\cdot+\cdot+} ))*_3( \adj( P_{\cdot +\cdot +}) P_{\cdot++\cdot} ) \right)\label{eq:bigflat}\end{equation}  are nonnegative, then $P$ arises from stochastic parameters.
\end{proposition}

The statement about the minors of the symmetric matrix in
\eqref{eq:bigflat} is of course really a requirement that this matrix be positive
semidefinite.  Also, this matrix could instead be replaced by ones
where the roles of leaves 1 and 2 or of leaves 3 and 4 have been
interchanged.

\begin{proof} Root $T$  at the interior node
near leaves $1$ and $2$. Let $M_i$, $i=1,2,3,4$ be the complex matrix parameter with row sums
equal to one on the edge leading 
to leaf $i$, $M_5$ the matrix parameter on the internal edge, and 
$\boldsymbol \pi$ the root distribution.  Define the matrices
\begin{align*}
N_{32}&=P_{+\cdot\cdot+}^T= M_3^TM_5^T\diag(\boldsymbol \pi) M_2, \\
N_{31}&=P_{\cdot+\cdot+}^T= M_3^TM_5^T\diag(\boldsymbol \pi) M_1.
\end{align*}
Then
$$
\overline P=P*_2N_{32}^{-1}N_{31}
$$
is a tensor arising from the same parameters
as $P$ except that $M_2$ has been replaced with $M_1$.  That is,
now the same matrix parameter is used on the edges leading to
leaves $1$ and $2$.

Similarly with 
\begin{align*}
N_{14}&=P_{\cdot++\cdot}= M_1^T\diag(\boldsymbol \pi) M_5 M_4,\\
N_{13}&=P_{\cdot+\cdot+}= M_1^T\diag(\boldsymbol \pi) M_5 M_3,
\end{align*}
then 
\begin{equation}\label{eq:Ptilde}
\widetilde P=\overline P*_3N_{13}^{-1}N_{14}
\end{equation} 
is a tensor arising from the same
parameters as $\overline P$ except that $M_3$ has been replaced with $M_4$.

Consider now the $13|24$ flattening of $\widetilde P$, a flattening 
which is \emph{not} consistent with the topology of the underlying tree. 
As shown in \cite{ARcovariontreeID2006}, this can be expressed as a 
product of $k\times k$ matrices
\begin{equation}\Flat_{13|24} (\widetilde P)= A^TDA,\label{eq:bigsym}\end{equation}
where $D$ is the diagonal matrix with the $k^2$ entries of
$\diag(\boldsymbol \pi)M_5$ on its diagonal, and $A=M_1\otimes M_4$ is
the Kronecker product.  Because $M_1,M_4$ are nonsingular, so is $A$.
Since conditions on $3$-marginals ensure $\boldsymbol \pi$ has
positive entries, we can ensure $M_5$ has nonnegative entries by
requiring that $\Flat_{13|24}( \widetilde P)$ be positive semidefinite.
Using Sylvester's theorem, this is equivalent to requiring that its principal minors
be nonnegative.  Since the resulting inequalities would
involve rational expressions, due to the inverses of matrices, we
first multiply $\Flat_{13|24} \widetilde P$ by squares of nonzero
determinants, to remove denominators.  \qquad\end{proof}

\smallskip

The matrix in \eqref{eq:bigflat} has entries of degree $4k+1$ in those of $P$.  
After removing squares of powers of determinants for
even minors, the polynomial inequalities from the principal $j \times j$ minors 
are of degrees
$j(2k+1) + 2ke_j$, where $e_j\in\{0,1\}$ gives the parity of $j$.

Together with Theorems \ref{thm:kstate-quadform} and \ref{thm:nonneg}, the last two propositions yield the following theorem.

\begin{theorem}\label{thm:kstate-semialgebraic}
  Suppose $P$ is an $n$-dimensional joint probability distribution for the $k$-state
  variables $Y_1,\dots, Y_n$.
  Then $P$ arises from nonsingular stochastic parameters for
  GM($k$) on an $n$-leaf binary tree $T$ if, and only
  if,

\begin{romannum}

\item All marginalizations of $P$ to $3$ variables satisfy the
  conditions of Theorem \ref{thm:kstate-quadform} (or if $k=2$ of
  Theorem \ref{thm:nonneg}) to arise from nonsingular stochastic
  parameters on a $3$-leaf tree;\label{itm:1}
   
\item For all internal edges $e$ of $T$, the edge invariants are satisfied by
   $P$, i.e., all $(k+1) \times (k+1)$ minors of the matrix flattening $Flat_e(P)$ are
   0;

\item For each internal edge $e$ of $T$, and some choice of $4$ leaves inducing a quartet
 tree with internal edge $e$, all principal minors of the matrix flattening constructed in
  Proposition \ref{prop:4taxa}, for the $4$-dimensional marginalization, are
  nonnegative.
\end{romannum}

\end{theorem}

\smallskip

While we noted that one can use a smaller set of inequalities than
were given in Proposition \ref{prop:nonsing3toN} to ensure a
distribution arises from nonsingular parameters, the full set of
inequalities given in Theorem \ref{thm:kstate-semialgebraic} has
additional redundancies.  To illustrate, in the $4$-leaf case checking
that only two of the $3$-marginals, say $P_{+\cdot\cdot\cdot}$ and
$P_{\cdot\cdot\cdot+}$ for the tree $12|34$, satisfy the conditions of
Proposition \ref{thm:kstate-quadform}.  is sufficient.

\smallskip

For a 4-variable distribution $P$, it is straightforward to obtain
semialgebraic conditions ensuring $P$ arises from strictly positive
parameters: One need only require the more stringent condition (iv')
of Proposition \ref{thm:kstate-quadform}, on the marginalizations
$P_{\cdot\cdot\cdot+}$ and $P_{+\cdot\cdot\cdot}$ to ensure they arise
from strictly positive parameters, and that all leading principal
minors of the matrix in \eqref{eq:bigflat} are strictly positive. This
then allows one to give such conditions applicable to larger trees, 
establishing the following.

\smallskip

\begin{theorem}  Semialgebraic conditions that a probability distribution 
$P$ arises from nonsingular positive parameters for \GMk on a tree $T$
can be explicitly given.
\end{theorem}

\smallskip

Note that one can also handle non-binary trees by the techniques of
this section. To show a distribution arises from nonsingular, or
stochastic nonsingular, parameters on a non-binary tree, one need only
show it arises from parameters on a binary resolution of the tree, and
that the Markov matrix on each edge introduced to obtain the
resolution is the identity. But semialgebraic conditions that the
Markov matrix on an internal edge of a 4-leaf tree be $I$ (or a
permutation, since label swapping prevents us from distinguishing
these) amounts to requiring that the matrix of equation
\eqref{eq:bigsym} has rank $k$.
Indeed, rank $k$ implies that the Markov
matrix on the internal edge has only $k$ nonzero entries, and since
other conditions we have derived imply nonsingularity, the matrix must
be a permutation.

\medskip

We now give an example illustrating that the quadratic form approach
of Proposition \ref{prop:4taxa}, and thus of Theorem
\ref{thm:kstate-semialgebraic}, detects a probability distribution
that is in the image of $\psi_T$ for nonsingular real \GMtwo
parameters on the $4$-taxon tree, where each matrix parameter on a
terminal edge is stochastic but the one on the internal edge is not.
By choosing parameters with some care, we can arrange that such a
probability distribution $P$ satisfies that all $3$-marginalizations
arise from stochastic parameters, yet $P$ does not.  Such examples are
not new (see for example \cite{ARquartets, KL2011, SZ2011}), but we
include one here to illustrate our methods.

To create such an example, set the Markov parameter on each terminal
edge to have positive entries, using, for instance, the same $M$ on
each of these 4 edges. Then choose the matrix parameter $N$ on the
internal edge of the tree to have very small negative off-diagonal
entries, so small so that both $MN$ and $NM$ are Markov matrices.  The
root distribution may be taken to be any probability distribution with
positive entries.  An example of such an (exact) probability
distribution is given by $P$ with slices
\begin{align*}\label{ex:4taxonInternalEdgeNonstochastic}
P_{\cdot \cdot 1 1} &= 
\left[
\begin{array}{cc}
   0.4005062  & 0.0565718\\
   0.0565718  & 0.0545702
\end{array}
\right] \hskip .5cm
P_{\cdot \cdot 1 2} = 
\left[
\begin{array}{cc}
   0.0457358   & 0.0141662\\
   0.0141662   & 0.0379118
\end{array}
\right] \\
P_{\cdot \cdot 2 1} &=
\left[
\begin{array}{cc}
   0.0457358  & 0.0141662\\
   0.0141662  & 0.0379118
\end{array}
\right] \hskip .5cm
P_{\cdot \cdot 2 2} =
\left[
\begin{array}{cc}
   0.0100222  & 0.0330958\\
   0.0330958  & 0.1316062
\end{array}
\right].
\end{align*}
Here $P$ satisfies all conditions  of 
Theorem \ref{thm:kstate-semialgebraic} except (iii).  A  computation shows
that the principal minors of the matrix
in \eqref{eq:bigflat} are, when rounded to eight decimal places,
$0.00363408$,  $ 0.00001744$,  $ 0.00000060$, and $ -0.00000005$.
The negativity of one of these shows $P$ does not arise from stochastic parameters.

\medskip

We conclude with a complete semialgebraic description of the
$2$-state general Markov model on a $4$-leaf tree without a restriction to 
nonsingular stochastic parameters.  This is straightforward to give, since 
Proposition \ref{prop:ind} indicates that in this case a distribution which arises 
from parameters either has independent leaf sets (so we can decompose the 
tree using Proposition  \ref{prop:patchtogetherindependentsets}), or the parameters 
were nonsingular so Theorem  \ref{thm:kstate-semialgebraic} applies.

As observed earlier, the existence of many non-independence 
cases when $k>2$ prevents us from assembling as complete a result.

\smallskip

\begin{proposition} For the $4$-leaf tree $12|34$, the image of the stochastic parameter space
  under the general Markov model $GM(2)$ is the union of the following sets of
  nonnegative tensors whose entries add to 1:

\begin{remunerate}

\item Probability distributions of $4$ independent variables: $P$ such that all $2\times2$
  minors of every edge flattening vanish (i.e., all edge flattenings have rank 1);

\item Probability distributions with partition into minimal independent sets of variables of
  size $1$, $3$, of which there are 4 cases: If the partition is $\left \{ \{Y_1\},\{Y_2,Y_3,Y_4\}\right \}$, then
  $P$ such that all $2\times 2$ minors of $\Flat_{1|234}(P)$ vanish, 
  and $P_{+\cdot\cdot\cdot}$ satisfies the conditions of Theorem
  \ref{thm:nonneg}, Case 1;

\item Probability distributions with partition into minimal
  independent sets of variables of size $1$, $1$, $2$, of which there
  are 6 cases: If the partition is $ \{Y_1\}|\{Y_2\}|\{Y_3,Y_4\}$,
  then $P$ such that all $2\times 2$ minors of $\Flat_{1|234}(P)$ and
  $\Flat_{2|134}(P)$ vanish, and $\det(P_{++\cdot\cdot})$ is nonzero;

\item Probability distributions with partition into minimal
  independent sets of variables $\left \{
    \{Y_1,Y_2\},\{Y_3,Y_4\}\right \}$ of size $2$, $2$: $P$ such that
  all $2\times 2$ minors of $\Flat_{12|34}(P)$ vanish, yet
  $\det(P_{\cdot\cdot++})$ and $\det(P_{++\cdot\cdot})$ are nonzero,

\item Probability distributions with no independent sets of variables: $P$ such that the edge invariants for
  $12|34$ are satisfied, the $3$-d marginalizations $P_{+\cdot\cdot\cdot}$ and
  $P_{\cdot\cdot\cdot+}$ satisfy the conditions of Theorem 
  \ref{thm:nonneg}, Case 1, and all
  principal minors of the matrix constructed in Proposition \ref{prop:4taxa} are
  nonnegative.

\end{remunerate}
\end{proposition}

\smallskip

In case 1, the only edge flattenings that are needed are those
associated to terminal edges. If these all have rank 1, then the
flattening for the internal edge does as well.

In cases 1,2,3, the distributions arise from stochastic parameters on
all 3 of the binary topological trees with 4 leaves, as well as the
star tree.

Note that all possible partitions of variables do not appear, but only
those consistent with the tree topology. In the 4-leaf case, this has
ruled out only the 2 partitions of size 2,2 that do not reflect a
split in the tree.

\smallskip

Of course one could extend the above proposition to arbitrary size
trees, as long as $k=2$, but the number of possible partitions
into independent sets of variables grows quickly, so we will not
give an explicit statement.

\bigskip

\section*{Acknowledgements}  We thank Piotr Zwiernik for his careful
reading of an earlier version of this paper and his helpful comments.  We
are also grateful to the Mittag Leffler Institute for its hospitality in early 2011,
where some of this work was undertaken.

\bibliographystyle{siam}

\bibliography{gmsemi}

\end{document}